\documentclass{lmcs}
\usepackage{amsmath}
\usepackage{amssymb}
\usepackage{euscript}
\usepackage{hyperref}
\usepackage{amscd, mathabx}
\usepackage{pstricks,pst-plot,pst-node}
\usepackage{color}
\usepackage{stmaryrd}

\begin{document}

\newcommand{\pfclm}[1]{\vspace{3mm}\noindent {\bf Proof of Claim\ #1}.}
\newcommand{\dom}{\mbox{dom\ }}
\newcommand{\ran}{\mbox{ran\ }}
\newcommand{\res}{\!\restriction\!}  
\newcommand{\<}{\langle}
\renewcommand{\>}{\rangle}
\newcommand{\nbhd}[1]{{\llbracket #1\rrbracket}}
\newcommand{\nbhdbar}[1]{\overline{{\llbracket #1\rrbracket}}}

\newcommand{\re}{r.e.~}
\newcommand{\CC}{{\mathcal C}}
\newcommand{\DD}{{\mathcal D}}
\newcommand{\M}{{\mathcal M}}
\newcommand{\N}{\mathbb{N}}
\newcommand{\Q}{\mathbb{Q}}
\newcommand{\Z}{\mathbb{Z}}
\newcommand{\R}{\mathbb{R}}
\newcommand{\C}{\mathbb{C}}
\newcommand{\sgn}{\operatorname{sgn}}
\newcommand{\F}{{\EuScript F}}
\newcommand{\D}{{\EuScript D}}
\newcommand{\G}{{\EuScript G}}
\newcommand{\J}{{\EuScript J}}
\newcommand{\baire}{{\mathcal N}}
\newcommand{\bairehex}{{\mathcal N}^{\#}}
\newcommand{\Ftilde}{\tilde{F}}
\newcommand{\DT}{\operatorname{DyaTree}}
\newcommand{\rs}[1]{\upharpoonright #1}

\newcommand{\PA}{\mathsf{PA}}
\newcommand{\concat}{\symbol{94}}

\newcommand{\Dy}{\mathrm{Dy}}
\newcommand{\Ds}{\mathrm{D}^*}

\title[A Foundation of Real Computation] {A Recursion Theoretic Foundation of Computation over Real Numbers}
\author{Keng Meng Ng, Nazanin R. Tavana and Yue Yang}
\address{Division of Mathematical Sciences, School of Physical \& Mathematical Sciences,
Nanyang Technological University, 21 Nanyang Link, Singapore 637371}
\address{Amirkabir University of Technology, 424 Hafez Ave, Tehran, Iran, P.O. Box: 15875-4413}
\address{Department of Mathematics, National University of Singapore, Block S17,
10 Lower Kent Ridge Road, Singapore 119076}

\email{kmng@ntu.edu.sg}
\email{nrtavana@aut.ac.ir}
\email{matyangy@nus.edu.sg}

\begin{abstract}
We define a class of computable functions over real numbers using functional schemes similar to the class of primitive and partial recursive functions defined by G\"odel \cite{Godel:1931, Godel:1934} and Kleene \cite{Kleene:1936}.
We show that this class of functions can also be characterized by master-slave machines, which are Turing machine like devices.
The proof of the characterization gives a normal form theorem in the style of Kleene \cite{Kleene:1936}.
Furthermore, this characterization is a natural combination of two most influential theories of computation over real numbers, namely,
the type-two theory of effectivity (TTE) (see, for example, Weihrauch \cite{Weihrauch:2000}) and the Blum-Shub-Smale \cite{Blum.Shub.ea:1989} model of computation (BSS).
Under this notion of computability, the recursive (or computable) subsets of real numbers are exactly effective $\Delta^0_2$ sets.
\end{abstract}

\thanks{The first author is partially supported by the MOE grant MOE-RG26/13.
The last author is partially supported by MOE grant MOE-2019-t2-2-121.}

\maketitle

\section{Introduction}
The original motivation of this paper was to understand the notion of algorithm in its general form, i.e.,
not necessarily over domain $\N$ --- the set of natural numbers\footnote{In this paper,
we use both $\N$ and $\omega$ to denote the set of natural numbers.}.
By studying algorithms over real numbers and making comparisons with their counterpart over $\N$,
we came across a notion of computability over real numbers---actually two similar notions,
one over Baire Space \footnote{In this paper, we use both $\baire$ and $\omega^{\omega}$ to denote the Baire space.}
and the other over $\R$ identified as the set of equivalence classes of Cauchy sequences.
The most interesting feature of this notion is its resemblance to the classical notion of ``computability over natural numbers''.
In fact, it can be viewed as a natural lift of the classical notion from many different aspects (recursive definitions, machine model
and definablility).
Before we elaborate the details, let us make some short comments on the study of algorithms and how it motivated us.

The concept of algorithm was formalized in 1930s. It is one of the greatest intellectual achievements in history.
In fact there are several equivalent formalizations based on different insights on effectiveness.
They all apply to the domain of natural numbers or objects which can be coded by natural numbers.
Furthermore, the formalizations themselves can be formulated within first-order Peano arithmetic.
All formulations captured the finite and discrete nature of algorithms,
and revealed the intrinsic link between computability and natural numbers.

However, the informal concept of algorithms applies not only the natural numbers but also to other domains.
For example, Newton's method of finding roots is an algorithm over real numbers and real-valued functions.
Algorithms are used in other areas of natural sciences too, for example, laboratory procedures for experimental scientists.
Most people would agree that (at least for a fixed domain) there is a clear intuition about what algorithms are.
What we hope to formalize is ``that clear intuition'' over real numbers.

Among the informal descriptions of algorithm, the one below gives the best illustration of our approach:

\begin{defi} [Informal] \label{informal}
An {\em algorithm} is a finite set of instructions such that
\begin{enumerate} 
\item each instruction is ``effective'' (that is, definite or clearly stated or sufficiently simple etc);
\item there is a clear organization of the instructions so that from input to output, we know how to go from one instruction to another.
\end{enumerate}
\end{defi}
\noindent What we like about this description is the separation of the instruction part and the overall organization part.
The instructions in item (1) may depend on domains and one needs to justify the effectiveness of them;
whereas the organization part in item (2) is really the heart of algorithms,
and it should be ``absolute'' when we move from one domain to another.

When the domain is the set of real numbers, part (1) actually depend on the underlying topology,
that is why we consider both the Baire space $\omega^{\omega}$ and the real numbers $\R$.
However, part (2) is done in the same way as over $\N$.
Over each of the domains, we formalize the notion of computability in two ways, namely, by functional schemes and by machines;
and show that the two formalizations actually give us the same class of functions.
We will come back to the discussion of algorithms in the concluding section.

\section{Computability Theory over Baire Spaces} \label{sec:Baire}

We choose Baire space $\baire=\omega^{\omega}$ as our first example,
because it will shed light on computation over real numbers while avoiding the tacky issues like dealing with equivalence classes of Cauchy sequences.
It is the first step of formalizing computability on higher types, an area closer to logicians than computational mathematicians.

As we observed earlier, computability and natural numbers are intrinsically linked together.
Therefore, we always need to have natural numbers at our disposal.
For this reason, when we talk about computability over a fixed domain $D$,
we consider the class of functions $\mathcal{F}$ from mixed types to mixed types, more precisely,
$\mathcal{F}$ consists of functions from $Y\to Z$ where $Y\in \{\N^p\times D^q, \N^p, D^q\}$ and $Z\in \{\N, D\}$,
where $p$ and $q$ are positive integers.
In particular functions from $\N$ to $\N$ and from $D$ to $D$ are included.
As functions with codomain $\N^p\times D^q$ can be viewed as juxtaposition of $p+q$ many functions in $\mathcal{F}$,
we will not study them explicitly.

For notational simplicity, we sometimes discuss only functions with domain $\N\times D$ instead of $\N^p\times D^q$.
Note that for domains like Baire space, $\N$ can be naturally embedded into them.
However, working with mixed types is necessary in more general domains, for example, computation on finite rings.
In the remaining of this section, the domain $D$ will be the Baire space $\baire$,
and $\mathcal{F}=\{f: f$ is a function from $\N^p\times \baire^q$ to either $\N$ or $\baire\}$.

\subsection{Formalizing computability using functional schemes}
We begin with a concept used by computable analysts, which can be traced back to the notion of relativised computation by Turing \cite{Turing:1939}.

\begin{defi}
We say that a partial function $F: \baire \to \baire$ is {\em TTE-computable} over $\baire$ if there is an oracle Turing machine $M$ such that
for all $x$ in $\baire$, $F(x)\downarrow =y$ if and only if for all $i$ in $\N$, $M^x(i)\downarrow=y(i)$;
and $F(x)\uparrow$ if and only if for some $i$ in $\N$, $M^x(i)\uparrow$.
(Here we used the standard notation that $M^x$ is the machine $M$ with $x$ as its oracle, and $\downarrow$ means ``is defined'' and $\uparrow$ means ``is undefined''.)
\end{defi}

One can also define TTE-computable functions as induced by recursive functions $f:\omega^{<\omega}\to \omega^{<\omega}$ satisfying the usual monotonicity conditions.
We use oracle Turing machines because of their direct connection with master-slave machines which we define later.
It is well-known that universal oracle Turing machine exists (see for example, Soare \cite{Soare:1987} p.48).  Thus we have

\begin{lem}
  There is a universal TTE-computable function $\Psi(e;x)$ over $\baire$, i.e., for any TTE-computable function $F:\baire \to \baire$,
  there is some $e\in \N$ such that for all $x\in \baire$, $F(x)=\Psi(e;x)$.
\end{lem}

Here is our first main definition:
\begin{defi} \label{pcX}
The class of {\em partial recursive functions over $\baire$} is the smallest subclass $\mathcal{C}$ of $\mathcal{F}$ satisfying the following conditions:
\begin{enumerate} [(1)]
\item $\mathcal{C}$ contains the following basic functions:
\begin{enumerate} [(a)]
 \item Zero function $Z: \N\to \N$, $Z(n)= 0_{\N}$;
 \item successor function $S: \N\to \N$, $S(n)= n+1$; and
 \item for natural numbers $p, q$ and $i$ with $p+q\geq 1$ and $1 \leq i\leq p+q$ the projection function
 \[
 \pi^{p+q}_i (n_1,\dots,n_p;x_1,\dots, x_q)
=\left\{
      \begin{array}{ll}
       n_i, & \hbox{if $i\leq p$;} \\
       x_{i-p}, & \hbox{if $i>p$.}
      \end{array}
 \right.
\]
\item A universal TTE-computable function $\Psi(e;x)$ over $\baire$;
and
\item the characteristic function $\chi:\baire\to \N$ of $\{0_{\baire}\}$ where $0_{\baire}$ is the constant zero sequence in Baire space.
\end{enumerate}
\item $\mathcal{C}$ is closed under
\begin{enumerate} [(a)]
\item composition, provided the types are matched;
\item primitive recursion with respect to $\N$; and
\item $\mu$-operator with respect to $\N$.
\end{enumerate}
\end{enumerate}
\end{defi}

The class of {\em primitive recursive functions over $\baire$} is defined similarly with the following changes:
(i) Fix a recursive list of indices $e_i, i=0,1\dots$ such that the Turing machine $M_{e_i}$ computes the $i$-th classical primitive recursive function.
Define the universal primitive TTE function $\Theta(i;x)$ as $\Theta(i;x):=\Psi(e_i; x)$.
Replace $\Psi$ in item (1)(d) by $\Theta$; (ii) Drop the clause (2)(c).
We say a predicate is primitive recursive over $\baire$ if its characteristic function is a primitive recursion function over $\baire$.
Observe that all primitive recursive over $\baire$ functions are total.

We now give precise definition of the terminologies used in the closure properties (item (2) in Definition \ref{pcX}).
In the definitions below, $\vec{n},\vec{m},\vec{x}$ and $\vec{y}$ stand for the tuples $(n_1,\dots,n_p)$, $(m_1,\dots,m_r)$,
$(x_1,\dots, x_q)$ and $(y_1,\dots,y_s)$ respectively.

\begin{defi} Given a function $f(n_1,\dots,n_p;x_1,\dots, x_q)$ and a finite sequence of functions  $g_i(\vec{m}; \vec{y})$ and $h_j(\vec{m}; \vec{y})$ where $1\leq i\leq p$ and $1\leq j\leq q$,
we say that the types of $f$, $g_i$ and $h_j$ are {\em matched} if the codomain of $g_i$ is $\N$ and the codomain of $h_j$ is $\baire$.

We say that a class of functions $\mathcal{C}$ is {\em closed under composition, provided the types are matched}, if for functions $f(\vec{n};\vec{x})$, $g_i(\vec{m}; \vec{y})$ and $h_j(\vec{m}; \vec{y})$ with matched types ($1\leq i\leq p$ and $1\leq j\leq q$), $f$, $g_i$ and $h_j$ are in $\mathcal{C}$ implies the function
\[
f(g_1(\vec{m}; \vec{y}),\dots, g_p(\vec{m}; \vec{y}); h_1(\vec{m}; \vec{y}), \dots, h_q(\vec{m}; \vec{y}))
\]
is in $\mathcal{C}$.
\end{defi}

\begin{defi}
Let $G:\N^p\times \baire^q\to \baire$ and $H:\N^{p+2}\times \baire^{q}\times \baire \to \baire$ be given.
We say that a function $F:\N^{p+1}\times \baire^q\to \baire$ is obtained from $G$ and $H$ {\em by primitive recursion} with respect to $\N$, if
\begin{eqnarray*}
  F(0,\vec{n};\vec{x})&=& G(\vec{n};\vec{x})\\
F(k+1,\vec{n};\vec{x})&=& H(k,\vec{n};\vec{x},F(k,\vec{n};\vec{x})).
\end{eqnarray*}
Similarly for function $F:\N^{p+1}\times\baire^q\to \N$ (just change the codomains of $G$ and $H$ to $\N$ and move the term $F(k,\vec{n}; \vec{x})$ before the semicolon in the second equation).
\end{defi}

\begin{defi}
Let $G:\N^{p+1}\times \baire^q\to \N$ be given.
We say that a function $F:\N^{p}\times \baire^q\to \N$ is obtained from $G$ {\em by $\mu$-operator} with respect to $\N$, if
\[
F(\vec{n};\vec{x})=\left\{
  \begin{array}{ll}
    \hbox{the least $k$ such that for any $k'<k$,}& \hbox{if such $k$ exists;}\\
    \hbox{$G(k',\vec{n}, \vec{x})\downarrow \neq 0_{\N}$ and $G(k,\vec{n},\vec{x})=0_{\N}$}, &  \\
     & \\
    \hbox{undefined}, & \hbox{otherwise.}
  \end{array}
\right.
\]
We write $F(\vec{n};\vec{x})$ as $\mu k\ G(k,\vec{n};\vec{x})=0$.
\end{defi}

Next we state a few lemmas which will be needed later.  Only sketches of the proofs are given, as they are almost entirely the same as the standard proofs.
Note that every classical recursive function on $\N$ has a natural pointwise lift on $\baire$.
For example, the lift of exponential function $n\mapsto 2^n$ is: $(n_0,n_1,\dots)\mapsto (2^{n_0}, 2^{n_1}, \dots)$.
We understand that this lift is not the usual exponential function on real numbers, that is why we need to go beyond Baire space later.

\begin{lem} \label{lift}
All partial recursive functions over natural numbers in the standard sense are partial recursive over $\baire$.
In particular, G\"odel numbering of finite sequences of natural numbers and the corresponding decoding functions are primitive recursive over $\baire$.
\end{lem}

\begin{proof}
By definition, it is clear that the generalized recursive function over $\baire$ contains the standard partial recursive functions as a subclass.
\end{proof}

\begin{lem} \label{equality}
The equality predicate ``$x=y$'' over $\baire$ is primitive recursive over $\baire$.
\end{lem}

\begin{proof}
By Lemma \ref{lift}, the cutoff subtraction $x\dotdiv y$ and the absolute value function $|x-y|=(x\dotdiv y)+(y\dotdiv x)$ are TTE-computable functions over $\baire$.
Therefore $x=y$ if and only if $\chi(|x-y|)=1$. The result follows.
\end{proof}

Lemma \ref{equality} explains the purpose of having the characteristic function $\chi$ of $\{0_{\baire}\}$: to make equality computable.

\begin{lem} \label{bounded_sum}
If $f:\N\times \baire\to \N$ is primitive recursive over $\baire$, then so are $\sum_{i=0}^n f(i,x)$ and $\prod_{i=0}^n f(i,x)$.
Consequently the class of primitive recursive over $\baire$ predicates is closed under bounded (numerical) quantifiers.
\end{lem}

\begin{proof}
The proof is exactly the same as the one in classical recursion theory.
\end{proof}

Lemma \ref{lift} and \ref{bounded_sum} also explains why we must include $\N$ in the formalizations of computability,
because we need the coding of finite sequences\footnote{This approach is similar to the prime computability introduced by Moschovakis
\cite{Moschovakis:1969}. We thank Alexandra Soskova for informing us the literatures in this area.}.

\begin{lem} [definition by cases] \label{cases}
  If $f_1(x)$ and $f_2(x)$ are primitive recursive functions over $\baire$, and $P(x)$ and $Q(x)$ are mutual exclusive primitive recursive
predicates over $\baire$, then the function
\[
f(x)=\left\{
  \begin{array}{ll}
    f_1(x), & \hbox{if $P(x)$ holds;} \\
    f_2(x), & \hbox{if $Q(x)$ holds.}
  \end{array}
\right.
\]
is also primitive recursive over $\baire$.
\end{lem}
%

\subsection{Formalizing Computability by Master-Slave Machines} \label{MSforBaire}
We now adopt the master-slave machines (abbreviated as MS-machines) used in Yang \cite{Yang:ta} to formalize the computability over $\baire$.
For the sake of simplicity, we will fully employ Church's Thesis of computation over natural numbers.

\begin{figure} \label{MSmachine}
\begin{center}
\begin{pspicture}(-8,-7.7)(8,7)
\rput(-2,0){\psset{unit=0.8}
\pspolygon(-1.1,4)(1.1,4)(1.1,2)(-1.1,2)
\rput(0,3){$M$}
\pspolygon(-3,-2)(-1.5,-2)(-1.5,-4)(-3,-4)
\rput(-2.25,-3){$S_0$}
\pspolygon(-1,-2)(0.5,-2)(0.5,-4)(-1,-4)
\rput(-0.25,-3){$S_1$}
\pspolygon(1,-2)(2.5,-2)(2.5,-4)(1,-4)
\rput(1.75,-3){$S_2$}
\pspolygon(2,0.5)(3,0.5)(3,-0.5)(2,-0.5)
\rput(2.5,0){$k$}
\psline(-5,0.5)(-1,0.5)(-1,-0.5)(-5,-0.5)(-5,0.5)

\psline[linestyle=dashed](-7,0)(7,0)

\rput(-3,0){Billboard}
\psline(-4,6)(4,6)
\psline(-4,5)(4,5)
\rput(0,5.5){tape for master}
\psline(-5,-5)(5,-5)
\psline(-5,-5.6)(5,-5.6)
\rput(0,-5.3){input from Baire space}
\psline(-5,-6)(5,-6)
\psline(-5,-6.6)(5,-6.6)
\rput(0,-6.3){zero test}
\pscurve{->}(5.5,-6.3)(6,-3)(3.2,-0.2)
\psline(-5,-7)(5,-7)
\psline(-5,-7.6)(5,-7.6)
\rput(0,-7.3){working and output tapes}
\rput(4,-3){$\cdots$}
\rput(5,-3){$\cdots$}

\psline{->}(0,4)(0,5)
\psline{->}(-0.5,2)(-3,0.5)
\psline{->}(0.5,2)(2.5,0.5)
\psline{->}(-2.25,-2)(-3,-0.5)
\psline{->}(-2.25,-4)(-4,-5)
\psline{->}(-0.25,-2)(-2.7,-0.5)
\psline{->}(-0.25,-4)(-1,-5)
\psline{->}(1.75,-2)(-2.4,-0.5)
\psline{->}(1.75,-4)(2,-5)
}
\end{pspicture}
\caption{A Master-Slave machine}
\end{center}
\end{figure}
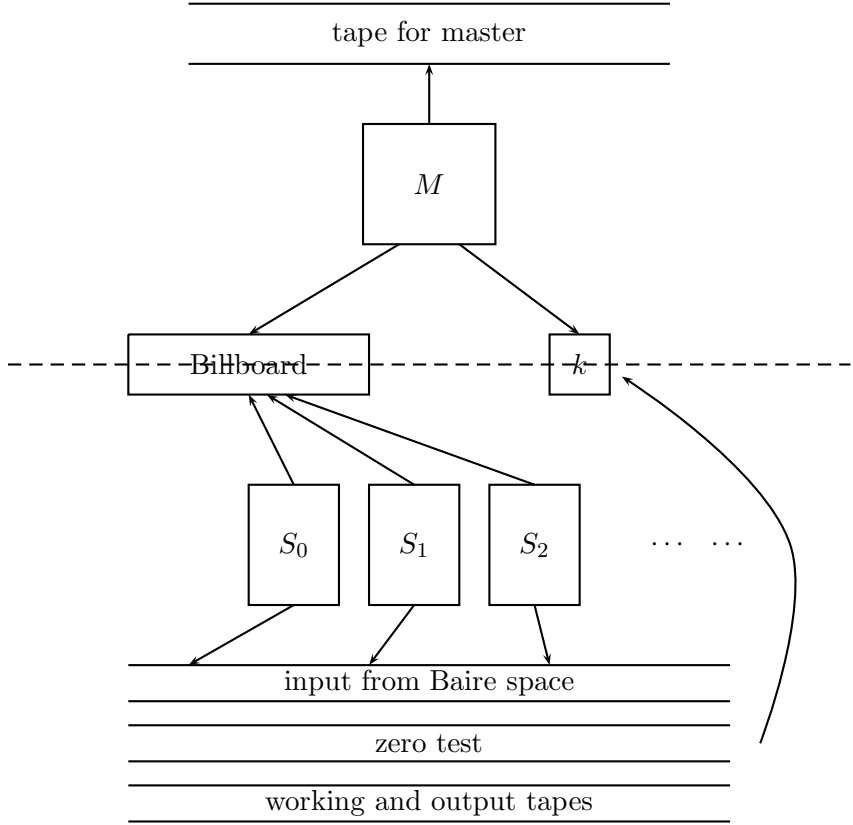

Let us quickly recall the basic features of a master-slave machine (see Figure \ref{MSmachine}).
Its physical device consists of the following three parts: Master part (above the dashed line in Figure \ref{MSmachine})
which only handles natural numbers; slave part (below the dashed line) which handles computations over Baire space;
and their interactions (on the dashed line) where the information flows from the master to slaves and vice versa.

The master part is simply a standard Turing machine $M$.  As it is the control part of the machine,
the same letter $M$ is also used for the whole MS-machine. If a computation involves natural number inputs,
they will be sent to the master $M$ as well.

Although one can view the slave part as a big black box, the description below gives more intuition to the computation,
and its effectiveness.  We have an infinite sequence $\<S_i: i\in \omega\>$ of universal (standard) oracle Turing machines,
called {\em slaves}.
Each slave $S_i$ has its own label $i$, indicating its relative position to other slaves.
For an element $x=(m_0,m_1,m_2,\dots)$ in Baire space, we put $x$ on tape as $b(m_0)\Box b(m_1) \Box b(m_2) \Box \dots$,
where $b(m)$ stands for $m$ written in binary and $\Box$ stands for empty cell.
All slaves share the following four kinds of tapes: (1) An oracle tape on which the input $x$ is written,
we sometimes refer to it as the {\em input tape} for slaves;
(2) one or more working tapes for slaves (besides their own working tapes for natural numbers);
(3) a write-only output tape for Baire space output, the slaves will write sequentially, in particular,
any slave cannot modify the writings of other slaves;
(4) a special tape, called {\em zero-test tape}, for any sequence written on this tape,
there is a special module which is able to detect if it is the zero sequence and to return a Boolean bit $k=0$ or $1$ to the master.

The interaction between the master and the slaves is done as follows.
The master writes a natural number $p$ on the {\em billboard} tape.
Each slave $S_i$ will take this $p$ as its program input and its index $i$ as numerical input.
Since the slaves are universal oracle Turing machines $\Phi$, what $S_i$ does is to simulate the TTE-computation $\Phi^x(p,i)$.
Here and below, we use $\Phi^x(p)$ for $\Psi(p;x)$ to emphasize that $x$ is an oracle; and $\Phi^x(p,i)$ is the $i$-th component of  $\Phi^x(p)$ if it is defined.)
The feedback from the slaves to master is done via zero-test module, and the result is stored in the Boolean bit $k$ which $M$ can use.

The program (or software part) of a master-slave machine can be identified with the Turing program of the master $M$.
We can ignore the slaves, because all slaves use the same fixed universal program.
Fix a finite set of states $Q=Q_0\cup\{S,E\}$ where $Q_0$ is the set of states for a standard Turing machine and
$S$ and $E$ are two new symbols for slave calling and zero-testing respectively.
We also single out two states $q_s$ and $q_h$ in $Q_0$ for starting and halting, respectively.
Also fix a set of alphabet $\Sigma$, usually binary and with some auxiliary symbols for convenience.
The transition functions of a master-slave machine can be represented by a finite set of finite tuples,
the dimension of the tuples depends on the number of tapes. For example,
the machine in Figure \ref{MSmachine} has three tapes (and three heads), namely, input tape, billboard and the Boolean bit.
However, in the discussion below, we will ignore the multi-tape issues for simplicity.
We will use only one tape, hence the transition function is described by quadruples.
Since any multi-tape standard Turing machine can be simulated by one with a single tape, we do not loss any generality here.
The quadruples are of the following three kinds: Standard, slave-calling and zero-testing.
\begin{enumerate} [(1)]
\item The standard commands are of the form $qaa'q'$ or $qaLq'$ or $qaRq'$,
where $q\in Q_0$, $q'\in Q$, $a,a'\in\Sigma$ and $L, R$ are for directions. Their executions are exactly the same as in classical case.

\item The slave calling commands are of the form $Saaq$, where $q\in Q$ and $a\in\Sigma$ (the symbol $a$ is just to make it a quadruple).
When executing this command, the slave $S_i$ will perform $\Phi^x(p,i)$ as described in the interaction part above.
When every slave machine halts, the state of the master becomes $q$.  If some slave does not halt,
the computation on input $x$ of the MS-machine is undefined.

\item The zero-test commands are of the form $Eaaq$ where $q \in Q$ and the symbol $a$ again is to make it a quadruple.
The execution is also described in the interaction part above.
\end{enumerate}

The definition of an MS-computation is analogues to the classical one. For simplicity again,
we assume that we only have one tape for master and one tape for slaves.
An {\em MS-configuration} $C$ is a pair $(n;y)$ where $n$ is a natural number coding the configuration (in the standard sense) of the master together with the billboard and Boolean bit;
and $y$ is an element in $\baire$ currently written on the tape handled by slaves.

Given a master-slave machine $M$ and a configuration $C=(n;y)$, we decode from $n$ in $C$ to get the current state parameter $q$ and the symbols $a$ reading by the master;
then check which quadruple in $M$ starts with $qa$ (in the case when $q$ is $S$ or $E$, we do not check the second component $a$) for some non-halting state $q$ and act according to the command, the resulting configuration $D$ will be the one {\em yielded} from $C$ by $M$.

A configuration $(n;y)$ is called {\em terminal}, if its state component coded in $n$ is the halting state $q_h$.
A terminal configuration does not yield any configurations.

For any input $(k;x)$ of mixed type,
the {\em initial MS-configuration} for input $(k;x)$ is the one that the Master is in its starting state $q_s$ and $k$ is put on the input tape of the master and the Baire space input $x$ is written on the oracle tape shared by all slaves.

A {\em master-slave computation} of $M$ on input $(k; x)$ is a sequence of MS configurations $\<C_i\>$ such that
$C_0$ is the initial configuration for $(k;x)$, and either
(1) the sequence is infinite and for all $i$, $C_i$ yields $C_{i+1}$; or
the sequence is of finite length $i_0+1$ for some natural number $i_0$, for all $i<i_0$, $C_i$ yields $C_{i+1}$,
and either (2.1) $C_{i_0}$ is a terminal MS-configuration, in which case,
the output is the number on the output tape of the master or slaves depending on the codomain;
or (2.2) the state component of $C_{i_0}$ is $S$ and $\Phi^y(p,j)\uparrow$ for some $j\in \omega$,
where $p$ and $y$ are the billboard and Baire space component of $C_{i_0}$ respectively.

\begin{defi}
We say that a partial function $f:\N^p\times {\baire}^q\to Y$ where $Y\in \{\N, \baire\}$ is {\em MS-computable} over $\baire$ if there is a master-slave machine $M$ such that
\[
f(\vec{n};\vec{x})=\left\{
       \begin{array}{ll}
         y, & \hbox{if $M$ on input $(\vec{n};\vec{x})$ halts and the output is $y$;} \\
            & \hbox{(note that $y$ can be in $\N$ or in $\baire$ depending on $f$.)}\\
         \uparrow, & \hbox{if (1) or (2.2) happens in the MS-computation of $M$}\\
          &\hbox{on input $(\vec{n};\vec{x})$.}
       \end{array}
     \right.
\]
\end{defi}

\begin{lem}
Every partial recursive function over $\baire$ is MS-computable over $\baire$.
\end{lem}

\begin{proof}
By Definition \ref{pcX}, it suffices to check that the basic functions are MS-computable over $\baire$;
and the class of MS-computable functions over $\baire$ are closed under composition, primitive recursion with respect to $\N$ and $\mu$-operator with respect to $\N$.

The basic functions in item 1(a) to 1(c) in Definition \ref{pcX} are clearly MS-computable.
The universal TTE-computable function is MS-computable: Given input $(e;x)\in \N\times \baire$,
the master just copies $e$ on billboard and calls the slaves.
The MS-computation of $\chi$ is simply done by the zero-test command.

The proof of the closure properties is similar to the ones in classical recursion theory, hence skipped.
\end{proof}

\subsection{A Normal Form Theorem for computation over $\baire$} \label{NormalforBaire}

Our definition of MS-computation has several quantifiers over $\N\times \baire$, it is not obvious how effective it is.
In this subsection, we analyze MS-computations carefully and define an analog of Kleene's $T$-predicate over $\baire$.
Consequently we get both the following Normal Form Theorem and the main characterization of MS-computable functions.

\begin{thm} [Normal Form Theorem] \label{KleeneNFT}
There are primitive recursive over $\baire$ predicate $T(e,x,z)$ and a (partial) TTE-computable function $U(z;x)$
where $e,z\in \N$ and $x\in \baire$, such that for any MS-computable function $F$ over $\baire$,
there is some $e$, for all $x\in \baire$
\[
F(x)=U(\mu z T(e,x,z);x).
\]
\end{thm}

\begin{thm} \label{main_X}
The class of MS-computable functions over $\baire$ coincides with the class of partial recursive functions over $\baire$.
\end{thm}

We now proceed to prove Theorem \ref{KleeneNFT}.
Let $F\in \mathcal{F}$ be MS-computable, say computed by an MS-machine $M$.
For notational simplicity, we assume that $F$ is from $\N\times \baire$ to $\baire$,
the cases of $\N^p\times \baire^q$ to $\baire$ or to $\N$ are similar.
And we will ignore the multi-tape issue again.
Since there is an effective way to list all MS-programs, this $M$ has an index, say $e$.
This $e$ will be carried along as a fixed parameter in the rest of the proof.

As in classical recursion theory, we need to ``arithmetize'' everything.
The first important feature is that this arithmetization is done almost entirely using natural numbers,
with minimal reference to the input $x$ in Baire space.
All transitions of configurations will be recorded by natural numbers only.
To do that, we introduce a new natural number parameter $c$, called the {\em TTE component of a configuration},
such that for input $x$, the sequence $\<\Phi^x(c,i): i\in \omega\>$ is the current tape content written on the slave's tape
(here we used the single tape convention).
In other words, we code a configuration $(n;y)$ as $(n,c; x)$ where $y(i)=\Phi^x(c,i)$ for all $i\in \omega$.
In particular, the initial configuration is coded as $(n_0,c_0;x)$ where $c_0$ is an index of the identity TTE function,
namely, $\Phi^x(c_0, i)=x(i)$ for all $i\in \omega$; and $n_0$ codes the starting state $q_s$ and the natural number input $k$, etc.
This idea is possible because the only change on the Baire space component is initiated by the slave calling command.
When we execute that command we get the content of billboard tape $p$ from $n$.
By $s$-$m$-$n$ theorem, there is a (standard) primitive recursive function $g(p,c)$ such that $\Phi^x(g(p,c))=\Phi^{\Phi^x(c)}(p)$.
In other words, all Baire space components that we ever see during the computation are TTE-images of the input $x$.

Introducing the TTE component is not enough because the master may call some partial TTE functions,
whereas we want to make the Kleene $T$-predicate primitive recursive, in particular $T$ must be total.
The idea is to make full use the infinitely many slaves and delay the detection of the $\Pi^0_2$ question of totality until the last step.

We have our second important (though a bit artificial) feature.  We add restrictions to the slave activities as follows.
Add to the slave alphabet a new symbol, say $\#$,
and restrict the running time of the $i$-th slave $S_i$ to $2^i$, in fact, the exact time bound is not important,
as long as the $i$-th slave can finish reading relevant part of the tape and computing for $i$ steps.
If the computation task of $S_i$ does not finish by step $i$, it writes a $\#$ and stop.
The unfinished task is left to the next slave.  Thus slave $S_{i+1}$ would start from checking the computations of $S_0, \dots, S_i$,
if it finds some unfinished computations, say $S_j$ is the first slave whose task is unfinished, $S_{i+1}$ will continue the task of $S_j$.
If within its time restriction $i+1$,
$S_{i+1}$ completed the task of $S_j$, it writes the result on tape and declare that $S_j$'s task is finished; otherwise,
it writes $\#$ and stop.

Let $\bairehex=(\omega\cup\{\#\})^{\omega}$.  For any element $x=(x(0),x(1),\dots)\in \baire$,
we say that $y\in \bairehex$ is a $\#$-{\em extension} of $x$, if either
\[
y=(\#^{n_0}, x(0), \#^{n_1}, x(1), \#^{n_2}, x(2), \dots)
\]
for some $(n_i)\in \baire$; or there is some $i\in \omega$ such that
\[
y=(\#^{n_0}, x(0), \#^{n_1}, x(1), \dots, \#^{n_i}, x(i),\#,\#,\dots).
\]
In the latter case, we say that $y$ is a {\em divergent} $\#$-extension with initial part $(x(0), \dots, x(i))$.
We also let $\Phi^x_{\#}(e)$ to denote the ``restricted'' universal oracle Turing machine.
Observe that $\lambda x. \Phi^x_{\#}(e):\bairehex \to \bairehex$ is total for every $e\in \N$.

It should be noted that by changing alphabet, we can easily embed $\bairehex$ into $\baire$ effectively, thus continuously.
Thus we could discuss everything within $\baire$ completely, the purpose of introducing $\bairehex$ is purely for intuitive clarity.

During the computation, the master will ``live'' in $\baire$, whereas the restricted slaves work in $\bairehex$
(except the last step of converting output).
For instance, suppose that the current configuration is $(n,c;x)$ and the master wants the slaves to computing $\Phi^y(p)$
where $y$ is the current tape content $\Phi^x(c)$.
However, the slaves now are working inside $\bairehex$, they typically see a $\#$-extension $y^{\#}$ of $y$ as input.
The slaves then compute $\Phi^{y^{\#}}_{\#}(p^{\#})$ where the program $p^{\#}$ will produce some $\#$-extension of $z$,
if $\Phi^y(p)\downarrow=z$.
Furthermore, we can get $p^{\#}$ effectively by $s$-$m$-$n$ Theorem.  We summarized these facts in the lemma below:

\begin{lem} \label{restricted_slave}
There is a primitive recursive function $g^{\#}$ such that for any $x\in \baire$, $p,c\in \omega$, we have
\begin{enumerate}
\item [(a)] if $\Phi^x(c)\downarrow=y$, $\Phi^y(p)\downarrow =z$, then $\Phi^x_{\#}(g^{\#}(p,c))$ is a nondivergent $\#$-extension of $z$;
\item [(b)] if $\Phi^x(c)\uparrow$ or ($\Phi^x(c)\downarrow=y$ but $\Phi^y(p)\uparrow$), then $\Phi_{\#}^{x}(g^{\#}(p,c))$ is a divergent $\#$-extension (with initial part being whatever $\Phi^{\Phi^x(c)}(p)$ produces).
\end{enumerate}
\end{lem}

On the other hand, when the master asks for a zero-test of $y$,
the slaves will execute some ``extended zero-test'' on some $\#$-extension of $y$, which treats the $\#$-symbols as zero.
Thus if the original zero-test gives answer $k$, the extended test will also give the same answer.

\begin{lem} \label{restricted_test}
  There is a function $\chi_{\#}: \bairehex\to \{0,1\}$, which is a composition of the zero-test function $\chi$ in $\baire$
  and some TTE over $\bairehex$ function, such that, for each $y^{\#}\in \bairehex$ which is a $\#$-extension of some $y\in \baire$,
$\chi_{\#}(y^{\#})$ gives the same answer as $\chi(y)$, more precisely,
  \[
  \chi_{\#}(y^{\#})=\left\{
                 \begin{array}{ll}
                   1, & \hbox{if $y=0_{\baire}$ or $y^{\#}$ is divergent with initial part all zero;} \\
                   0, & \hbox{if $y\neq 0_{\baire}$  or $y^{\#}$ is divergent with non-zero initial part.}
                 \end{array}
               \right.
\]
\end{lem}

With these two features in mind, the rest of the proof is parallel to the classical proof of Kleene's Normal Form Theorem.
Define the function $t((n,c;x),e)=((m,d;x),e)$ reflecting the fact ``$C$ yields $D$ by $M$'' via their codes as follows:

First we decode from $n$ to get the current state $q$ and the current symbol read by the master, say $a$;
then check which command or quadruple in $M$ (which is coded by $e$) starts with $qa$ (in the case when $q$ is $S$ or $E$,
we do not check the second component $a$).
Depending on the command, we obtain the code of the resulting configuration $D$ as follows.
\begin{enumerate} [(1)]
\item If the quadruple is $qaa'q'$ or $qaLq'$ or $qaRq'$, then the transition only affects the master part, say it changes $n$ to $m$.
Then in this case, $t((n,c;x),e)= ((m,c;x),e)$. Note that $n\mapsto m$ is primitive recursive over $\baire$ by Lemma \ref{lift}.

\item If the quadruple is $Saaq$, where $S$ is the slave calling state, then we can get the content of billboard tape $p$ from $n$.
Let $g^{\#}(p,c)$ be the function defined in Lemma \ref{restricted_slave}.
Set $t((n,c;x),e)= ((m,g^{\#}(p,c);x),e)$ and $m$ is the resulting code of the configuration of the master part to record the change of state to $q$.
Note that $(n,c)\mapsto (m,g^{\#}(p,c))$ is again primitive recursive over $\baire$

\item If the quadruple is $Eaaq$ where $E$ is the special zero-test state,
then the state of master changed from $E$ to $q$ and the boolean bit becomes $\chi_{\#}(\Phi_{\#}^x(c))$,
where $\chi_{\#}$ are the function defined in Lemma \ref{restricted_test}.
Let $m$ be the new code of the master configuration.
we have $t((n,c;x),e)= ((m,c;x),e)$. Note that $n\mapsto m$ is again primitive recursive over $\baire$.
Also note that it is this step that the input $x$ is used.
\end{enumerate}
Apply the definition by cases (Lemma \ref{cases}) and by discussions above, the transition function $t$ is primitive recursive over $\baire$.

Clearly to determine whether a code $(n,c;x)$ of is one of a terminal configuration is primitive recursive:
Just check if the state component $q$ in $n$ is the halting state $q_h$ of $M$.

Define the {\em Kleene T-predicate} $T(e,x,z)$ by ``$z$ is a (natural number) code of finite sequences $\<(n_i, c_i): i\leq |z|\>$ such that
$(n_0,c_0;x)$ is the initial configuration and for all $i< |z|$, $t((n_i,c_i;x),e)=((n_{i+1}, c_{i+1};x),e)$ and $(n_{|z|}, c_{|z|}; x)$ is terminal''.
Note: Although $z$ is a natural number, we did not put $z$ in front of $x$ because we want to keep it the same form as in classical recursion theory.

\begin{lem}
$T(e,x,z)$ is primitive recursive over $\baire$.
\end{lem}

\begin{proof}
It follows from that coding input, transition $t$ and deciding terminal configuration are all primitive recursive over $\baire$;
and primitive recursive predicates are closed under bounded quantification (Lemma \ref{bounded_sum}).
\end{proof}

Finally, by applying $\mu$-operator, we can find the least $z$ such that $T(e,x,z)$, if exists.
(If the function has codomain $\N$, then the output reading function $u: \N\times \N\times \baire\to\N$ is exactly as in classical case.)
We define the output reading functions $U: \N\times \baire\to \baire$ as follows:
Given a code of a finite sequence of configurations $z$ and the Baire space input $x$,
we first get from $z$ the set of all slave calling commands that $M$ ever used, say the indices are $p_1, \dots, p_r$,
From $i=0$, we read through from $j=1$ to $r$, the $i$-th non-$\#$-component of $\Phi_{\#}^x(p_j)$.
If any of the $\#$-extension is divergent, then $U(z;x)$ is undefined, otherwise,
we write the $i$-th non-$\#$-component of $\Phi_{\#}^x(p_r)$ sequentially, that is the value of $U(z;x)$.
Clearly, $U$ is a partial TTE-computable function.

That finishes the proof of the Kleene Normal Form Theorem, consequently, we have the characterization theorem (Theorem \ref{main_X}).

\subsection{Characterizing the computable/recursive sets} \label{sec:characterization}
Recall that the basic open set in Baire space is of the form $\nbhd{\sigma}=\{f\in \baire: \sigma\prec f\}$ where $\sigma\in \omega^{<\omega}$.
We use $\nbhdbar{\sigma}$ to denote the closed set $\baire\setminus \nbhd{\sigma}$.

\begin{defi} \label{effectively_open}
A subset $O\subseteq \baire$ is said to be {\em effectively open} or $\Sigma^0_1$ over $\baire$
if there is a (classically) recursive sequence $\<\sigma_i\>_{i\in \omega}$
such that $O=\bigcup_i \nbhd{\sigma_i}$. We say that $C\subseteq \baire$ is {\em effectively closed} or $\Pi^0_1$ over $\baire$
if $\baire\setminus C$ is effectively open.

A subset $A\subseteq \baire$ is said to be an {\em effectively $G_\delta$ set} or $\Pi^0_2$,
if there is a recursive sequence $\<\sigma_{i,j}\>$ such that $A=\bigcap_i \bigcup_j \nbhd{\sigma_{i,j}}$.
$A$ is an {\em effectively $F_\sigma$ set} or $\Sigma^0_2$ if its complement is effectively $G_\delta$.
$A$ is $\Delta^0_2$ over $\baire$ if it is both effectively $G_{\delta}$ and effectively $F_{\sigma}$.
\end{defi}

\begin{defi}
We say that a subset $A\subseteq \baire$ is {\em recursive over Baire space}
if its characteristic function $\chi_{_A}$ is partial recursive over $\baire$.
\end{defi}

Clearly, from its definition, recursive over $\baire$ sets are closed under complementation.

\begin{lem} \label{lem:pi01isstronglycomputable}
Every effectively open subset $O$ of $\baire$ is recursive over $\baire$, so is every effectively closed set.
\end{lem}

\begin{proof} Let $O=\bigcup_i \nbhd{\sigma_i}$ where $\<\sigma_i\>_{i\in \omega}$ is a recursive sequence.
We design an MS-machine $M$ to compute $O$ as follows: Given input $x\in \baire$,
the master of $M$ writes the code of computing $\<\sigma_i\>$ on the billboard;
and asks the $i$-th slave $S_i$ to check if $\sigma_i\prec x$, i.e., $x\in \nbhd{\sigma_i}$;
if so, $S_i$ writes a $1$ on the $i$-th cell of the zero-test tape; it writes a $0$ otherwise.
Next we use zero-test to get the boolean bit $k$.  If $x$ is in $O$ then some slave will write a $1$ on the zero-test tape,
so the sequence being tested is non-zero, so $k=0$;
if $x$ is not in $O$, then all slaves write $0$ on the zero-test tape, so $k=1$.  In other words $x\in O$ if and only if $k=0$.
Note that $M$ actually decides the membership of $O$ in two master steps, one slave calling and one zero-test.
\end{proof}

For example, the Cantor space $2^{\omega}$ is recursive over $\baire$ as a subset of the Baire space.

\begin{lem} \label{lem:delta02isstronglycomputable}
If a subset $A$ of Baire space is $\Delta^0_2$ over $\baire$, then $A$ is recursive over $\baire$.
\end{lem}

\begin{proof} Since $A$ is $\Delta^0_2$, both $A$ and its complement are $\Sigma^0_2$.
Let $A=\bigcup_i\bigcap_j \nbhdbar{\sigma_{i,j}}$ and $\baire\setminus A=\bigcup_i\bigcap_j (\nbhdbar{\tau_{i,j}})$,
where $\<\sigma_{i,j}\>$ and $\<\tau_{i,j}\>$ are recursive sequences in the classical sense.
Suppose $x\in \baire$ is the input for the MS-machine $M$ which we are designing.
Using the algorithm described in the proof of Lemma \ref{lem:pi01isstronglycomputable},
$M$ can decide if $x$ is in the closed set $\bigcap_j (\nbhdbar{\sigma_{0,j}})$ using two master steps.
If $x$ is in such a closed set, then $M$ outputs $1$ indicating that ``$x$ is in $A$''.
Otherwise, $M$ will try to decide if $x\in \bigcap_j (\nbhdbar{\tau_{0,j}})$,
if yes, output $0$ indicating that ``$x$ is in $\baire\setminus A$''.
Otherwise, try $\bigcap_j (\nbhdbar{\sigma_{1,j}})$, then try $\bigcap_j (\nbhdbar{\tau_{1,j}})$, etc.
Since $x$ is either in $A$ or in $\baire\setminus A$.
This algorithm always halts and will give correct answer.
\end{proof}

The converse of Lemma \ref{lem:delta02isstronglycomputable} is also true.
The proof depends on another analysis of how MS-machine works, on top of the one given in the proof of the Normal Form Theorem.

Let $M$ be an MS-machine. Firstly we make all slaves in $M$ restricted as described in subsection \ref{NormalforBaire},
thus push the partialness issue to the very end.
We associate to $M$ a binary tree $\Gamma:=\Gamma(M)$ , called the {\em computation tree} for $M$, as follows.

Roughly speaking, $\Gamma$ simply lists all possible $M$-computations in tree form.
Each node $\sigma$ in $\Gamma$ is associated with two kinds of parameters:
One is a (code of a) configuration as defined in the Normal Form Theorem;
the other is a triple of indices $(i,j,k)=(i(\sigma), j(\sigma),k(\sigma))$,
where $i$ and $j$ are indices for some effective $\Sigma^0_1$ and $\Pi^0_1$ sets $O_i$ and $C_j$ respectively;
and $k$ is index for the TTE function $\lambda x. \Phi^x_{\#}(k)$.
It intuitively means that at $\sigma$, we restrict our attention only to the set $O_i\cap C_j$ which is determined by zero-tests,
and $\Phi^x_{\#}(k)$ is applied to $x \in O_i\cap C_j$ which has produced the contents on the tape(s) for slaves.

The root of $\Gamma$ is associated with the initial configuration with empty input,
and is labelled $(i_0, j_0, k_0)$ where $i_0$ and $j_0$ are indices of the whole space $\bairehex$,
and $k_0=c_0$ which is the index of identity function defined in the proof of Normal Form Theorem.
Suppose that we have defined the associated configuration and the indices at $\sigma$.
We can decode from the configuration to get the state component $q$.
Depending on $q$ we have the following possible cases:

Case 1. $q$ is the halting state. Then declare $\sigma$ is a terminal node on $\Gamma$.

Case 2. $q$ is the slave calling state $S$. Then $\sigma$ has a unique successor $\sigma\concat 0$.
We associate the `yielded' configuration to $\sigma\concat 0$.
If the number written on the billboard for slave to execute is $p$, $k(\sigma\concat 0)=g^{\#}(p,k(\sigma))$,
where $g^{\#}$ is defined in Lemma \ref{restricted_slave}.  Set $i(\sigma\concat 0)=i(\sigma)$, and $j(\sigma\concat 0)=j(\sigma)$.

Case 3. $q$ is the zero-test state $E$. Then $\sigma$ has two successors $\sigma\concat 0$ and $\sigma\concat 1$.
At node $\sigma\concat 0$ (which indicates the test result is `No', i.e., the sequence being tested has symbols other than zero and $\#$),
we associate the `yielded' configuration (whose boolean bit is $0$ now) to $\sigma\concat 0$.
If the content of the zero-test is prepared via TTE-function $\Phi_{\#}(z)$, we can effectively find the index $i^*$
for the $\Sigma^0_1$-set $\Phi_{\#}(z)^{-1}[\bairehex\setminus\{0,\#\}^{\omega}]$, and index $i^{**}$ for $O_{i(\sigma)}\cap O_{i^*}$.
Set $i(\sigma\concat 0)=i^{**}$. Set $j(\sigma\concat 0)=j(\sigma)$ and $k(\sigma\concat 0)=k(\sigma)$.
Do the similar things at $\sigma\concat 1$ (just replace $O, i$ by $C,j$ respectively).

Case 4. $q$ is one of the standard states. Then $\sigma$ has a unique successor $\sigma\concat 0$.
Read from the configuration associated at $\sigma$ the current command $qaa'q'$ and the current scanned symbol $b$.
At node $\sigma\concat 0$ we associate the corresponding `yielded' configuration to $\sigma\concat 0$.
Leave all indices unchanged.
%

This finishes the construction of the computation tree $\Gamma$.
Clearly $\Gamma$ is a recursive tree with $\sigma\mapsto (i(\sigma),j(\sigma),k(\sigma))$ recursive.

From this computation tree, we can get several consequences.

\begin{thm} \label{thm:MScomputable_equal_Delta2} A partial function $F:\baire\to \baire$ is MS-computable if and only if
there are a (classical) recursive set $R\subseteq \omega$ and a recursive function $\alpha:R\to \omega^3$
mapping $\sigma\mapsto (i,j,k')$ satisfying the following conditions:
\begin{enumerate}
  \item [(a)] The sets $O_{i(\sigma)}\cap C_{j(\sigma)}$ are mutually disjoint for $\sigma\in R$.
  \item [(b)] If $F(x)$ is defined, then there is a $\sigma\in R$ such that $x\in O_{i(\sigma)}\cap C_{j(\sigma)}$
  and $F(x)=\Phi^x(k')$.
  \item [(c)] If $F(x)$ is undefined, then either $x\notin \bigcup_{\sigma\in R} (O_{i(\sigma)}\cap C_{j(\sigma)})$
  or $x\in O_{i(\sigma)}\cap C_{j(\sigma)}$ for some $\sigma\in R$ but $\Phi^x(k')$ is undefined.
\end{enumerate}
\end{thm}

\begin{proof}
  Suppose that $F$ is computed by the MS-machine $M$, let $\Gamma$ be its computation tree described above.
  Define $R=\{\sigma\in \Gamma:$ the state component of the configuration associated with $\sigma$ is halting$\}$.
(Here we identify the node $\sigma$ on $\Gamma$ with its G\"odel number.)
  Since $\Gamma$ is recursive, so is $R$. Let $\alpha$ map $\sigma$ to $(i, j, k')$, where $(i,j,k)$ is the label of $\sigma$ in $\Gamma$
  and $k'$ is the index of the TTE computable function $U\circ \lambda x. \Phi^x_{\#}(k)$,
where $U$ is the output reading function defined in the Normal Form Theorem.
  To see that (a) is satisfied, any branching at $\Gamma$ can only be caused by Case 3 in the definition of $\Gamma$.
  The split is caused by zero-test and the resulting sets are certainly disjoint.

For (b) if $F(x)$ is defined, we will hit a halting node $\sigma$ on $\Gamma$,
  so $\sigma\in R$ and all slave calls along the computation path up to $\sigma$ are not partial,
so $\Phi^x(k)$ would produce the correct result $F(x)$.  Similarly (c) holds.

Conversely, suppose that we have recursive $R$ and $\alpha$, we design an MS-machine $M$ as follows.
For any input $x$, $M$ checks through every element $\sigma\in R$ to see if $x\in O_{i(\sigma)}\cap C_{j(\sigma)}$.
For each fixed $\sigma$, this can be done with finitely many master steps as in the proof of Lemma \ref{lem:pi01isstronglycomputable}.
If no such $\sigma$ is found, then $M(x)\uparrow$, however by (b), $F(x)$ also diverges.
If such $\sigma$ is found, then $M$ just call the TTE functional $\Phi^x(k')$, by (b) and (c), $M(x)=F(x)$.
\end{proof}

\begin{cor}
  The domain $W$ of an MS-computable function is of the form $\bigcup_n [(O_{i_n}\cap C_{j_n})\cap Y_n]$,
where $Y_n$ is $\Pi^0_2$ subset $\{x: \Phi^x(k_n)\in $ domain of $U\}$, and $n\mapsto (i_n,j_n,k_n)$ is primitive recursive.
\end{cor}

\begin{thm} \label{thm:stronglycomputableisdelta02}
If a subset $A$ of Baire space is recursive over $\baire$, then $A$ is $\Delta^0_2$ over $\baire$.
\end{thm}

\begin{proof}
Let $M$ be an MS-machine computing the characteristic function $\chi_{_A}$ and $R$ be the recursive set defined in
Theorem \ref{thm:MScomputable_equal_Delta2}.
Since $M$ is total, for every $\sigma\in R$, the domain of $\lambda x. \Phi^x(k')$ has to be a superset of $O_{i(\sigma)}\cap C_{j(\sigma)}$,
where $k'$ is also as defined in Theorem \ref{thm:MScomputable_equal_Delta2}.
Thus we have $x\in A$ if and only if for some $\sigma\in R$ with $x\in O_{i(\sigma)}\cap C_{j(\sigma)}$ and the natural number output
read from the configuration associated with $\sigma$ is $1$.
Consequently, $A$ is $\Sigma^0_2$. Similarly the complement of $A$ is also $\Sigma^0_2$.  So $A$ is $\Delta^0_2$ over $\baire$.
\end{proof}

This result shows that recursive sets are closer to the intuitively computable sets.
For instance, they do not include all Borel sets, in fact, not all arithmetical sets.

For subsets of natural numbers, we have a similar characterization.

\begin{thm} \label{char_on_N}
A subset $A$ of natural numbers is MS-computable if and only if $A\leq_T\emptyset'$, i.e., it is a $\Delta^0_2$-subset of natural numbers.
\end{thm}

\begin{proof}
Observe that the natural number input is directly sent to the master,
whereas the input tape for slaves are empty at the beginning of the computation.

Suppose that $A$ is MS-computable, say by the MS-machine $M$.
For any input $n$, we trace the computation tree $\Gamma$ associated with $M$.
At each branching point $\sigma$ cause by the zero-test, since the content prepared for the zero test is by
$\Phi^{\emptyset}(k(\sigma))$ where $k$ is defined in Theorem \ref{thm:MScomputable_equal_Delta2}.
Furthermore, since $\chi_{_A}$ is total, $\Phi^{\emptyset}(k(\sigma))$ is total.
Hence whether the sequence is the zero-sequence is a $\Pi^0_1$ question:
\[
\forall n\forall s(\Phi^{\emptyset}(k(\sigma),n)\downarrow[s]\rightarrow \Phi(k(\sigma),n)=0).
\]
So $\emptyset'$ can produce the answer.
Consequently, $\emptyset'$ can trace $\Gamma$ until the terminal node and read off the answer of $A(n)$.  So $A\leq_T\emptyset'$.

On the other hand, if $A\leq_T\emptyset'$ say via the oracle Turing machine $\Phi^{\emptyset'}(e)$.
The master can mimic $\Phi(e)$ until it reaches a query step asking if some number $n$ is in $\emptyset'$.
At that moment, the master employs the slaves to figure it out and read the answer from the boolean bit.
Thus, $A$ is MS-computable.
\end{proof}

\subsection{Relation with BSS models}

In \cite{Blum.Shub.ea:1989}, L.~Blum, Shub and Smale introduced a highly influential model (BSS machines) of computation over real numbers,
in fact, over arbitrary rings.
BSS machines treat a real number or an element in Baire space as a single mathematical entity,
and only focus on functions which are mostly common in scientific computing, such as polynomials or rational functions.
To compare our model with BSS machines, we have to make some reasonable modifications to both models.
First of all, BSS machines can use arbitrary real numbers as parameters, in particular, there are uncountably many BSS machines.
We have to restrict BSS model so that it does not use real parameters.

Now let us discuss computability on (ordered) rings with three basic ring operations $+$, $-$ and $\times$.
We modify MS-machines as follows:
\begin{enumerate}
\item We assume that we have only three slaves to compute the basic ring operations respectively.
Thus by calling them finitely many times, we can compute any polynomial in $\mathbb{Z}[\vec{x}]$.
\item We assume the zero-test module can tell whether a nonzero element is positive or negative too.
\end{enumerate}

With these modification, we have
\begin{thm}
  Over a ring $R$, the modified MS-machines compute exactly the same class of functions as the (parameter free) BSS-machines.
\end{thm}

\begin{proof} (Sketch)
Define the class of {\em partial recursive functions} over an ordered ring $R$ as the smallest class $\mathcal{C}$ of functions on $\N^p\times R^q$ satisfying the following conditions:
\begin{enumerate} [(1)]
\item $\mathcal{C}$ contains the following basic functions:
\begin{enumerate} [(a)]
\item Zero function $Z:\N\to \N$ with $Z(n)=0$;
 \item successor function $S: \N\to \N$ with $S(n)=n+1$; and
 \item for natural numbers $p, q$ and $i$ with $p+q\geq 1$ and $1 \leq i\leq p+q$ the projection function
 \[
 \pi^{p+q}_i (n_1,\dots,n_p;x_1,\dots, x_q)
=\left\{
      \begin{array}{ll}
       n_i, & \hbox{if $i\leq p$;} \\
       x_{i-p}, & \hbox{if $i>p$.}
      \end{array}
 \right.
\]
\item The primitive functions in ring $R$, namely $+_R$, $-_R$ and $\times_R$;
and

\item The sign function $\sgn:R\to \N$ defined by
\[
\sgn(x)=\left\{
  \begin{array}{ll}
    1, & \hbox{if $x>0_R$;} \\
    0, & \hbox{if $x=0_R$;} \\
    2, & \hbox{if $x<0_R$.}
  \end{array}
\right.
\]
\end{enumerate}
\item $\mathcal{C}$ is closed under
\begin{enumerate} [(a)]
\item composition, provided the types are matched;
\item primitive recursion with respect to $\N$; and
\item $\mu$-operator with respect to $\N$.
\end{enumerate}
\end{enumerate}

By mimic the proof of the Normal Form Theorem, one can show that the modified MS-machines compute exactly the partial recursive functions over $R$.
On the other hand, by Theorem on page 33 in \cite{Blum.Shub.ea:1989} section 7, BSS machines compute exactly the same class.
\end{proof}

On the other hand, let us look at computability on the Baire space.
First of all, we delete from BSS model the function of taking additive inverse, as we focus on functions on $\baire$, instead of on rings.
Since BSS machines as defined originally does include functions beyond polynomials, for example,
certain sets\footnote{There are examples over real numbers (see Brattka \cite{Brattka:2003}); but can be adapted over Baire spaces.} involving exponential function is not computable in BSS sense.
It is necessary to add functions that we are interested in, for instance, exponential function.
Obviously, we do not want to add them ``in the hardware'' i.e., merely allow those functions to
appear in the computing nodes of BSS machines, otherwise, we merely move up one level of primitive recursive hierarchy
and we will miss, say, tower of exponential functions.
Instead, we should add them through universal functions, so that the machine can call them dynamically.
Thus, we add the universal TTE-computable function $\Phi(e;x)$ in the BSS model, so that it can call any TTE-computable function.

\begin{thm} \label{thm:BSSandMS}
  Over the Baire space $\baire$, the modified BSS-machines compute exactly the same class of functions as the MS-machines.
\end{thm}

\begin{proof} (Sketch) By slightly modifying the proof of Theorem on page 33 in \cite{Blum.Shub.ea:1989} section 7,
modified BSS-machines can compute all partial recursive functions over $\baire$.
On the other hand, Turing machines can mimic all computations defined by flow-charts,
thus MS-machines can compute all functions that modified BSS-machine can compute.
The theorem then follows from Theorem \ref{main_X}.
\end{proof}

Theorem \ref{thm:BSSandMS} tells us (informally) that over $\baire$, MS-computation model is the smallest one that includes both BSS and TTE computation.
The reasons are: MS-machine is certainly more powerful than BSS and TTE.
On the other hand, if a computation model includes both BSS and TTE,
then it would be more powerful than the modified BSS machine. By Theorem \ref{thm:BSSandMS}, it is more powerful than MS-machine.

\section{Computability Theory on Real Numbers}

\subsection{Formalize Computability over Real Numbers}
We now discuss the computability over real numbers $\R$. Although Baire space is often viewed as a representation of $\R$,
computationally and topologically, $\baire$ is quite different from ``the real numbers'' that ``working mathematicians'',
like numerical analysts, have been working on.
We will make our meaning precise in the discussion below.

First of all, by well-known effective codings, we can identify the sets of integers $\Z$
and rational numbers $\Q$ with the set of natural numbers $\N$, respectively.
In particular, we will not study negative numbers explicitly. Again by effective coding,
we can identify the arithmetical operations of the field $\Q$, the ordering on $\Q$ and functions like absolute value function,
as primitive recursive functions over $\N$.
Through this coding, infinite rational sequences will be identified with infinite sequences of natural numbers.
We use $\Q^{\omega}$ to denote the set of all infinite rational sequences.
It is worth noting that the coding is not a homeomorphism between the Baire space and $\Q^{\omega}$ equipped with the topology induced by the usual metric on $\Q$.  Therefore, we will use $\Q^{\omega}$ instead of $\baire$.

We fix the following presentation of real numbers:

\begin{defi}
\begin{enumerate} [(1)]
\item A {\em fast converging Cauchy sequence} is a sequence of rational numbers $\<r_i:i\in \omega\>$ such that for all $m<n$
    $|r_m-r_n|<2^{-m-1}$.  Remark: In the remaining of the paper, since we will focus only on fast converging Cauchy sequences, whenever we say Cauchy sequences we mean fast converging ones.
\item Define a relation $\sim$ between Cauchy sequences by $\<r_i\>\sim \<s_i\>$ if for any $n$ there is some $m>n$ such that $|r_m-s_m|<2^{-n}$ (hence $\lim_i r_i=\lim_i s_i$).
\item Observe that $\sim$ defined in (2) is an equivalence relation. We identify the $\sim$-equivalence classes as real numbers.
The set of real numbers is denoted by $\R$.
\end{enumerate}
\end{defi}

In this section we use $i,j,k$ to denote natural numbers, $\alpha,\beta$ to denote infinite Cauchy sequences,
and $x,y$ to denote real numbers.

A (finite) rational sequence $\sigma=\<q_0, \dots, q_{n-1}\>\in\Q^{<\omega}$ is said to be \emph{Cauchy} if for every $i<j< n$ we have $\left|q_i-q_j\right|<2^{-i-1}$.
Given $\sigma=\<q_0, \dots, q_{n-1}\>,\tau=\<r_0,\dots, r_{m-1}\>\in\Q^{<\omega}$,
we say that $\sigma$ and $\tau$ are \emph{Cauchy compatible} or simply \emph{C-compatible} if $|q_{n-1}-r_{m-1}|<2^{-n}+2^{-m}$.
The point of C-compatibles strings is that they can be extended to $\sim$-equivalent infinite Cauchy sequences; whereas the non C-compatible ones cannot.
In other words, if $\alpha\supset \sigma$ and $\beta\supset\tau$ and $\alpha,\beta\in\Q^\omega$ and $\sigma$ and $\tau$ are not C-compatible, then $\lim \alpha\neq \lim \beta$;
equivalently, if $\lim \alpha=\lim \beta$ then for every $i,j$ we have that $\alpha\rs{i}$ and $\beta\rs{j}$ are C-compatible.

We say that a function $F:\Q^{\omega}\to \Q^{\omega}$ is {\em Cauchy preserving}\footnote{This and other notions in this section are studied
extensively by the school of computable analysis, see Chapter 3 in \cite{Weihrauch:2000}. That is why we use the phrase ``TTE-computable''.
However, since we only study the computability on $\R$, we did not introduce the general naming systems, admissibility etc.
Our approach and theirs are essentially similar, but may be different in some details.} if $F$ is only defined on Cauchy sequences and
if two Cauchy sequences $\alpha$ and $\beta$ are $\sim$-equivalent and $F(\alpha)\downarrow$ then $F(\alpha)$ is a Cauchy sequence,
$F(\beta)\downarrow$ is also a Cauchy sequence and $F(\alpha)\sim F(\beta)$.
Clearly, each Cauchy preserving function $F:\Q^{\omega}\to \Q^{\omega}$ naturally induces a function $\Ftilde$ from $\R$ to $\R$.

\begin{defi} \label{dfn:TTE_R}
We say that a function $\Ftilde: \R\to \R$ is {\em TTE-computable} over $\R$,
if there exists a TTE-computable function $F:\Q^{\omega}\to \Q^{\omega}$
which is Cauchy preserving and $\Ftilde$ is induced by $F$.
\end{defi}

This definition suggests the two-step approach towards computability on $\R$, the first is to have an algorithm working on a representative
$\alpha$, and the second is to check that this algorithm is Cauchy preserving.
The second step is the difference between computability over $\baire$ and over $\R$.
This is more pertinent when we talk about machines.
Over $\baire$, the input of the TTE-computable functions and the object written on the input tape of a machine are the same;
whereas over $\R$, the input of the functions are equivalence classes which cannot be written on the input tape.
A typical oracle Turing machine is sensitive to the digits on the oracle tape, hence will not be Cauchy preserving at all.
The main issue is to ``fine-tune'' the machine so that its operations are independent of the representation of the equivalence class.
Once the fine-tuning is done, the rest is similar to Section \ref{sec:Baire}.

We say that an oracle Turing machine $\widehat{M}$ is {\em fine-tuned} if it satisfies the following two conditions:
(1) If $\alpha\in \Q^{\omega}$ is not Cauchy, then $\widehat{M}^{\alpha}$ is undefined;
(2) if $\alpha$ and $\beta$ are $\sim$-equivalent Cauchy sequences,
then $\widehat{M}^{\alpha}$ is defined if and only if $\widehat{M}^{\beta}$ is defined and
the outputs $\widehat{M}^{\alpha}$ and $\widehat{M}^{\beta}$ are also $\sim$-equivalent Cauchy sequences.

We describe an effective procedure which transfers an oracle Turing machine $M$ into a fine-tuned one.
By this transfer, the universal oracle Turing machine will become a fine-tuned universal oracle Turing machine,
which naturally induces a universal TTE-function over $\R$.

The plan is to associate a ``canonical'' object to a real number $x$, which is invariant under its representations.
One candidate for the canonical object is ``the'' binary expansion of $x$.
Unfortunately, it is not effective to extract a unique binary expansion for $x$, for example, when $x$ is a dyadic rational,
there are two such expansions.
Recall that a dyadic rational is a rational number whose denominator is a power of $2$.
To overcome this problem, for a given Cauchy sequence $\alpha$, we effectively produce a binary tree $T_{\alpha}$
whose infinite paths are exactly the binary expansion(s) of $x=\lim \alpha$;
namely, if $x$ is not a dyadic rational, then $T_{\alpha}$ has a unique infinite path which is the binary expansion of $x$;
if $x$ is dyadic then $T_{\alpha}$ has exactly two paths which are the two binary expansions of $x$.
If $\alpha$ and $\beta$ are Cauchy equivalent, then $T_{\alpha}$ and $T_{\beta}$ may not be the same,
but they will have the same set of infinite paths.

The algorithm that transfer a Cauchy sequence $\alpha$ to $T_{\alpha}$ goes as follows.
For each natural number $n$, let $D_n$ denote the dyadic rational number with denominator $2^n$,
i.e. $D_n=\{\frac{m}{2^n}: m\in \Z\}$ and $D=\bigcup_{n\in \N} D_n$ be the set of all dyadic rational numbers.
For natural number $n$ and integer $m$, let
\[
J_n^m=\left[\frac{m}{2^n}, \frac{m+1}{2^n}\right]
\]
which is a closed interval of length $\frac{1}{2^n}$.
Given a Cauchy sequence $\alpha$, let
\[
I_n(\alpha)=\left[\alpha(n+4)-\frac{1}{2^{n+2}}, \alpha(n+4)+\frac{1}{2^{n+2}}\right],
\]
so that $I_n(\alpha)$ has the following properties:
\begin{enumerate}
\item [(1)] $I_n(\alpha)$ is a closed interval of length $\frac{1}{2^{n+1}}$.
\item [(2)] Since $\alpha$ is Cauchy, $x=\lim \alpha$ is in $I_n(\alpha)$.
\item [(3)] $I_{n+1}(\alpha)\subseteq I_n(\alpha)$.
\end{enumerate}

We now define the binary tree $T_{\alpha}$ in terms of $I_n:=I_n(\alpha)$ and $J^m_n$.
Let $\emptyset$ be the root of $T_{\alpha}$.  The recursive construction of $T_{\alpha}$ and the labelling process are as follows.
Since $I_0$ is of length $\frac{1}{2}$, it can intersect at most two intervals of the form $J_0^m$ for some $m$.
If $I_0$ intersects two such intervals, say $J_0^{m_0}$ and $J_0^{m_0+1}$ for some $m_0\in \Z$,
then the root has two children labelled from left to right as $d^0_L=\frac{m_0}{2^0}$ and $d^0_R=\frac{m_0+1}{2^0}$ respectively;
if $I_0$ intersects only one such interval, say $J_0^{m_0}$ for some $m_0\in \Z$,
then the root has only one child labelled $d^0=\frac{m_0}{2^0}$.

Suppose that we have defined the tree $T_{\alpha}$ up to level $\leq \ell$.
First check if the finite sequence $\alpha\rs (\ell+4)$ is Cauchy, if it is not, then stop the construction;
otherwise, proceed as follows.  If $T_{\alpha}$ has a single node of length $\ell$,
then proceed as the base case, except replacing $0$ by $\ell+1$.
If $T_{\alpha}$ has two nodes of length $\ell$, say $\rho_L$ and $\rho_R$,
labelled by dyadic rationals $d^{\ell}_L, d^{\ell}_R\in D_{\ell}$ respectively,
where $d^{\ell}_R=d^{\ell}_L+\frac{1}{2^{\ell}}$.
Find the interval $I_{\ell+1}$.  Since it has length $\frac{1}{2^{\ell+2}}$, it can intersect at most two intervals of the form $J_{\ell+1}^m$ for some $m$.

Case 1. $I_{\ell+1}$ intersects two such $J$-intervals.  Since the $I_{\ell+1}\subseteq I_{\ell}$, the $J$-intervals have to be $[d^{\ell}_L+\frac{1}{2^{\ell+1}}, d^{\ell}_R]$ and $[d^{\ell}_R, d^{\ell}_R+\frac{1}{2^{\ell+1}}]$.
We extend $\rho_L$ to a single node labelled $d^{\ell}_L+\frac{1}{2^{\ell+1}}$ and $\rho_R$ to a single node labelled $d^{\ell}_R$.

Case 2. $I_{\ell+1}$ intersects only one such $J$-interval.  Again since the $I$-intervals are nested,
the $J$-interval has to be either $[d^{\ell}_L+\frac{1}{2^{\ell+1}}, d^{\ell}_R]$ or $[d^{\ell}_R, d^{\ell}_R+\frac{1}{2^{\ell+1}}]$.
If it is former, extend $\rho_L$ to a single node labelled by $d^{\ell}_L+\frac{1}{2^{\ell+1}}$ and declare $\rho_R$ a dead end;
if it is latter, extend $\rho_R$ to a single node labelled by $d^{\ell}_R$ and declare $\rho_L$ a dead end.
That finished the construction of the tree $T_{\alpha}$.

By construction, $T_{\alpha}$ has the following properties:
\begin{enumerate}
\item [(1)] $T_{\alpha}$ is an infinite binary tree if and only if $\alpha$ is Cauchy.
\item [(2)] Assuming that $\alpha$ is Cauchy and $x=\lim\alpha\not\in D$,
then $T_{\alpha}$ has a unique path $\delta=\delta(T_{\alpha})$ such that for each $\ell$,
$\delta(\ell)$ is the largest dyadic in $D_{\ell}$ which is less than $x$.
\item [(3)] Assuming that $\alpha$ is Cauchy and $x=\lim\alpha\in D_n$ for the least such $n$,
then for all $\ell\geq n$, $\delta_R(T_{\alpha})(\ell)=x$ and $\delta_L(T_{\alpha})(\ell)=x-\frac{1}{2^{\ell}}$.
\item [(4)] The function from $\alpha\rs(\ell+4)\mapsto T_{\alpha}\rs \ell$ is recursive.
\end{enumerate}
It follows from (2) and (3) that if two Cauchy sequences $\alpha$ and $\beta$ are equivalent then $T_{\alpha}$ and $T_{\beta}$ have the same paths.

With the help of $T_{\alpha}$, we can describe an effective procedure which transfers an oracle Turing machine $M$ into a fine-tuned $\widehat{M}$.
The basic idea is that we monitor the (at most two) infinite paths $\delta_L$ and $\delta_R$ of $T_{\alpha}$,
$\widehat{M}$ simulates $M^{\delta_L}$ and $M^{\delta_R}$ and we accept the value only when the results are C-compatible.
The details are as follows:

Given $M$, $n$ and $\alpha$, the fine-tuned machine $\widehat{M}$ first generates the tree $T_{\alpha}$ (by property (4) above,
this can be done), and simulates $M^{\delta_L}(n)$ and $M^{\delta_R}(n)$,
where $\delta_L$ and $\delta_R$ are the paths through $T_{\alpha}$ (they might be equal) until one of the following happens:
  \begin{enumerate}
    \item [(a)] There is a natural number $s$, for all $i\leq n$, $M^{\delta_L\rs{s}}(i)\downarrow=r_i$ and $\<r_i:i\leq n\>$ is a finite Cauchy sequence and $\delta_L\rs{s}$ is the only node of length $s$ that still have extensions on $T_{\alpha}$.  Set the output to be $r_n$.
    \item [(b)] There is a natural number $s$, for all $i\leq n$, $M^{\delta_L\rs{s}}(i)\downarrow=r_i$,
    $M^{\delta_R\rs{s}}(i)\downarrow=t_i$ and $\<r_i:i\leq n\>, \<t_i:i\leq n\>$ are finite Cauchy sequences,
    $|r_i-t_i|\leq \frac{1}{2^i}$ and both $\delta_L\rs{s}$ and $\delta_R\rs{s}$ still have extensions on $T_{\alpha}$.
Set the output to be $\frac{r_n+t_n}{2}$.
  \end{enumerate}
In case no such $s$ exists, the procedure produces no output for this set of input $M, n,\alpha$.

Observe that if $\widehat{M}^{\alpha}(n)\downarrow=a_n$ for every $n$, then $\<a_n: n\in \N\>$ is a Cauchy sequence.
When $M^{\alpha}=\beta$ is a Cauchy sequence with $y=\lim \beta$, $|a_n-\beta(n)|<\frac{1}{2^n}$.  Let $b_n=a_{n+2}$ we will have
\[
|b_n-y|=|a_{n+2}-y|\leq |a_{n+2}-\beta(n+2)|+|\beta(n+2)-y|\leq \frac{1}{2^{n+2}}+\frac{1}{2^{n+2}}<\frac{1}{2^n}.
\]
Consequently, $\widehat{M}^{\alpha}$ and $M^{\alpha}$ are Cauchy equivalent.  Finally the procedure of getting $\widehat{M}$ from $M$ is effective.

We summarize the discussion as follows:

\begin{lem} \label{finetuningR}
  Let $\<\lambda \alpha. \Phi^{\alpha}(e): e\in \omega\>$ be an effective list of all TTE computable functions over $\Q^{\omega}$.
  Then there is a total recursive function $g:\N\to \N$ such that $\lambda \alpha. \Phi^{\alpha}(g(e))$ is Cauchy preserving.
  Furthermore, if $\lambda \alpha. \Phi^{\alpha}(e)$ is Cauchy preserving itself, then $\lambda \alpha. \Phi^{\alpha}(g(e))$ and $\lambda \alpha. \Phi^{\alpha}(e)$ induce that same function from $\R$ to $\R$.
\end{lem}

We verify the expected fact\footnote{\label{footnote6}Various versions of Lemmas \ref{lem:TTE_continuous} and \ref{lem:finding_index_for nonzero} are well-known for computable analysts.
For example, Lemma \ref{lem:TTE_continuous} may follow from Theorem 4.3.1 in Weihrauch \cite{Weihrauch:2000} and passing through effective quotient topology.
For the sake of completeness and correctness under our setting, we include outlines of the proofs.} for TTE-computable functions with respect to the usual topology:

\begin{lem} \label{lem:TTE_continuous}
  Any TTE-computable function $f$ over $\R$ is continuous.  To be more precise,
suppose that $f:\R\to\R$ is computed by the fine-tuned MS-machine $M$, and $f(x_0)\downarrow=y_0$, then for any $n>0$,
  there is an $m>0$ such that for any Cauchy sequences $\alpha$ with $\lim \alpha=x_0$ and for any Cauchy sequence $\beta$ with $\lim\beta=x$, if $M^{\beta}\downarrow$ and $|x_0-x|<\frac{1}{m}$, then $|M^{\beta}-y_0|<\frac{1}{n}$.
\end{lem}

\begin{proof}
Given $\frac{1}{n}>0$, choose $i$ such that $\frac{1}{2^i}<\frac{1}{2n}$.  Since $f(x_0)\downarrow$, for any $\alpha$ with limit $x_0$, $M^{\alpha}\downarrow$.  Choose $j$ such that $M^{\alpha\restriction j}(i)\downarrow$, by definition,
$|M^{\alpha\restriction j}(i)-y_0|<\frac{1}{2^i}$.  Now, by the fine-tuning, $M$ will first produce $T_{\alpha}$.
At the stage $s$ when $M^{\alpha\restriction j}(i)[s]\downarrow$, there are two possible cases.  Case 1. $T_{\alpha}$ has a unique node
of length $j$ that is still having extensions on $T_{\alpha}$ at stage $s$. By fine-tuning, the closed interval
$I_j(\alpha)=[a,b]$ is a proper subinterval of some dyadic interval $J_j^{m_0}=[c,d]$, thus $c<a<b<d$.
Let $m$ be such that $\frac{1}{m}<\min\{\frac{a-c}{2}, \frac{d-b}{2}\}$.  Case 2. $T_{\alpha}$ has a two nodes
of length $j$ that are still having extensions on $T_{\alpha}$ at stage $s$.  By fine-tuning, the the closed interval
$I_j(\alpha)=[a,b]$ intersects two consecutive dyadic intervals $J_j^{m_0}=[c,c']$ and $J_j^{m_0+1}=[c',d]$, thus $c<a\leq c'\leq b<d$.
Let $m$ be such that $\frac{1}{m}<\min\{\frac{a-c}{2}, \frac{d-b}{2}\}$. In both cases, if a Cauchy sequence $\beta$ satisfying
$|x_0-\lim\beta|<\frac{1}{m}$, the node(s) of length $j$ that is still having extensions on $T_{\beta}$ are also on $T_{\alpha}$.
Thus $M^{\beta\restriction j}(i)=M^{\alpha\restriction j}(i)$.  If $M^{\beta}\downarrow$, then it is a Cauchy sequence, and
\[
|M^{\beta}-y_0|<|M^{\beta}-M^{\beta\restriction j}(i)|+|M^{\beta\restriction j}(i)-y_0|<\frac{1}{2^i}+|M^{\alpha\restriction j}(i)-y_0|<\frac{1}{n}.
\]
\end{proof}

\begin{thm} [Enumeration Theorem for TTE-computable functions]
  There is a universal function $\Psi(e;x):\N\times \R\to \R$ for TTE-computable over $\R$, i.e.,
  for any TTE-computable function $\Ftilde:\R \to \R$, there is some $e\in \N$ such that for all $x\in \R$, $\Ftilde(x)=\Psi(e;x)$.
\end{thm}

\begin{defi} \label{pcReal}
The class of {\em partial recursive functions over $\R$} is the smallest subclass $\mathcal{C}$ of $\mathcal{F}$ satisfying the following conditions:
\begin{enumerate} [(1)]
\item $\mathcal{C}$ contains the following basic functions:
\begin{enumerate} [(a)]
 \item Zero function $Z: \N\to \N$, $Z(n)= 0_{\N}$;
 \item successor function $S: \N\to \N$, $S(n)= n+1$; and
 \item for natural numbers $p, q$ and $i$ with $p+q\geq 1$ and $1 \leq i\leq p+q$ the projection function
 \[
 \pi^{p+q}_i (n_1,\dots,n_p;x_1,\dots, x_q)
=\left\{
      \begin{array}{ll}
       n_i, & \hbox{if $i\leq p$;} \\
       x_{i-p}, & \hbox{if $i>p$.}
      \end{array}
 \right.
\]

\item A universal TTE-computable functions $\Psi(e;x)$ over $\R$; and
\item the characteristic function $\chi:\R\to \N$ of $\{0_{\R}\}$ where $0_{\R}$ is the real number zero.
\end{enumerate}
\item $\mathcal{C}$ is closed under
\begin{enumerate} [(a)]
\item composition, provided the types are matched;
\item primitive recursion with respect to $\N$; and
\item $\mu$-operator with respect to $\N$.
\end{enumerate}
\end{enumerate}
\end{defi}

We say that an MS-machine is fine-tuned, if it only writes index $g(e)$ on its billboard where $g$ is the function defined in Lemma \ref{finetuningR}.

\begin{defi}
We say that a partial function $f:\N^p\times \R^q\to \R$ is {\em MS-computable} over $\R$  or simply MS-computable if there is a fine-tuned
master-slave machine $M$ such that
\[
f(\vec{n};\vec{x})=\left\{
       \begin{array}{ll}
         y, & \hbox{if for any representation $(\vec{n};\vec{\alpha})$ of input,}\\
            & \hbox{$M^{\vec{\alpha}}(\vec{n})\downarrow=\beta$ and $\beta$ is a representation of $y$;}\\
            & \\
            & \hbox{if for any representation $(\vec{n};\vec{\alpha})$ of input,}\\
         \hbox{undefined},  &  \hbox{$M^{\vec{\alpha}}(\vec{n})$ never halts; or}\\
             & \hbox{some slave calls are partial during}\\
            & \hbox{the compuatation.}
       \end{array}
     \right.
\]
We define $f:\N^p\times \R^q\to \N$ being {\em MS-computable} similarly.
\end{defi}

\begin{lem} \label{lem:zero_test_R}
The characteristic function of $\{0_{\R}\}$ is MS-computable over $\R$.
\end{lem}

\begin{proof}
This is because all Cauchy sequences are fast converging. Given $\alpha$ which is any presentation of $x$,
we may ask the slaves $S_i$ to check if
\[
\forall j<i (|\alpha(j)|<\frac{1}{2^{j-1}})
\]
and write a zero if the answer is yes, write $\frac{\alpha(i)}{2}$ otherwise.
Observe that this preparation function is TTE and Cauchy preserving.
Then the usual zero-test would determine if $x=0_{\R}$.
\end{proof}

\subsection{Properties of Computable Sets and Functions over $\R$}

We will establish some basic results on computability over $\R$, many of which have corresponding ones over the Baire space $\baire$.
Although the statements look similar, we actually cannot directly transform the results from $\baire$ to $\R$, because of the step 2
mentioned in the remarks after Definition \ref{dfn:TTE_R}. For example, the maps
\[
((x(0), x(1),\dots, ),(y(0),y(1), \dots)\mapsto (x(0),y(0),x(1),y(1),\dots)
\]
is continuous from $\baire$ to $\baire$, but the corresponding map from $\Q^{\omega}\times \Q^{\omega}\to \Q^{\omega}$ is not Cauchy preserving.

Also note that since BSS machines completely ignore the issue of representations, we will not compare with BSS models over $\R$.

\subsubsection{Normal Form Theorem}
As in section \ref{sec:Baire}, we have the following Normal Form Theorem and characterization theorem over $\R$.

\begin{thm} [Normal Form Theorem for $\R$] There are TTE computable over $\R$ function $U(z;x)$ and a primitive recursive predicate $T(e,x,z)$ over $\R$ such that
for any partial recursive function $F$ over $\R$, there is some $e$, for all $x\in \R$
\[
F(x)=U(\mu z T(e,x,z),x).
\]
\end{thm}

\begin{thm}
The class of MS-computable functions over $\R$ coincides with the class of partial recursive functions over $\R$.
\end{thm}

As the proofs are similar to the one in Baire space, we only check the issues related to Cauchy preserving step.
When executing the slave calling command, the function $g^{\#}(p,c)$ does not depend on the real input $x$.
When executing the zero-test command, the preparation step, i.e., writing the sequence being tested on zero-test tape,
is done by ``fine-tuned'' slaves, hence is Cauchy preserving (see the proof of Lemma \ref{lem:zero_test_R}).
For the issue of working in $\bairehex$, one can do the same for the space $\Q^{\omega}$ with the extra symbol$\#$.
For instance, one can define a sequence (finite or infinite) with $\#$ symbols {\em $\#$-Cauchy} if it is Cauchy after
removing the $\#$ symbols, and fine-tune the machines to make them ``$\#$-Cauchy preserving'', instead of ``Cauchy preserving''.
The same proof will go through.

\subsubsection{Characterizing the computable sets over $\R$} \label{sec:characterization_R}
Recall that the basic open set in $\R$ is of the form $B(c;r)=\{x\in \R: |x-c|<r\}$ where $c,r\in \Q$ and $r>0$.

\begin{defi}
A subset $A\subseteq \R$ is said to be {\em effectively open} or $\Sigma^0_1$ over $\R$ if there is a (classical) computable sequence $\<e_i\>_{i\in \omega}$
such that $A=\bigcup_i B(c_i,r_i)$ where $e_i$ is the code of the pair $(c_i,r_i)$. We say that $B\subseteq \R$ is {\em effectively closed} or $\Pi^0_1$ over $\R$ if $\R\setminus B$ is effectively open.

{\em Effectively $G_\delta$ set} or $\Pi^0_2$, {\em effectively $F_\sigma$ set} or $\Sigma^0_2$ and $\Delta^0_2$ over $\R$ sets can be defined
in a similar fashion as in Definition \ref{effectively_open}.
\end{defi}

By Lemma \ref{lem:TTE_continuous}, any TTE-computable function $F$ is continuous,
hence the pre-images $F^{-1}[\{x\in \R: x\neq 0_{\R}\}]$ and $F^{-1}[\{0_{\R}\}]$ are open and closed respectively.
We verify that their indices can be effectively found\footnote{Similar results are done by computable analysts, see footnote \ref{footnote6}.}.
Let $\varphi_e(i)$ denote the standard universal partial recursive function.

\begin{lem} \label{lem:finding_index_for nonzero}
There is a recursive function $g: \N\to\N$ such that for any TTE-computable over $\R$ function $\lambda x. \Phi^x(e)$, the set
\[
\{x\in \R: \Phi^x(e)\neq 0\}=\bigcup_{i\in \N} B(c_{\varphi_{g(e)}(i)}, r_{\varphi_{g(e)}(i)}),
\]
hence is open.
\end{lem}

\begin{proof}
Let $M_e$ be the MS-machine that computes $\lambda x. \Phi^x(e)$.
The following uniform effective procedure gives us the open set from the index $e$, whose code will be $g(e)$:

Step 1. For each $n\in \Z$, define a recursive tree $R_n$ as follows: The root of $R_n$ is $n$.
Suppose that $\sigma\in R_n$ is defined and is of length $s$, then let $\sigma\concat (\sigma(s))$ and $\sigma \concat (\sigma(s)+\frac{1}{2^{s+1}})$ be the successors of $\sigma$ on $R_n$.  Clearly, each $R_n$ satisfies the following properties:
(1) For each $\sigma\in R_n$, $\sigma$ is a finite nondecreasing dyadic Cauchy sequences with $\sigma(j)\in [n,n+1]\cap D_j$
for every $j<|\sigma|$;
(2) for any $x\in [n,n+1]$ there is an infinite path $\delta$ of $R_n$ such that $\lim \delta=x$.
By simultaneously enumerate $(R_n: n\in \Z)$, we are in fact enumerating the trees that contain all paths on
$T_{\alpha}$ for all Cauchy $\alpha\in \Q^{\omega}$, where $T_{\alpha}$ is the tree produced by the fine-tuning procedure.

Step 2. For each $\sigma\in R_n$, let $s=|\sigma|$.  Mimic the computation $M_e^{\sigma}$ for $s$ steps.
If for some $j<s$, $M_e^{\sigma}(j) \downarrow[s]$ and is not C-compatible with $0^{s}$,
then enumerate the interval $(\sigma(s), \sigma(s)+\frac{1}{2^s})$ into $O$.  (The first two steps take care of nondyadic $x$.)

Step 3. For each dyadic rational number $d$, let $\delta_L^d=(d-\frac{1}{k}:k\in \N)$ and $\delta_R^d$ be the constant squence $(d,d,\dots)$.
We can enumerate all finite initial segments of $\delta_L^d$ and $\delta_R^d$ for all possible $d$.
For each $d$, let $\delta_L:=\delta_L^d$ and $\delta_R:=\delta_R^d$.
Mimic the computation $M_e^{\delta_L\restriction s}$ and $M_e^{\delta_R\restriction s}$ for $s$ steps.
If for some $j<s$, $M_e^{\delta_L\restriction s}(j) \downarrow[s]$,  $M_e^{\delta_R\restriction s}(j)\downarrow[s]$ and both are Cauchy and
their average is not C-compatible with $0^{j}$,
then enumerate the interval $(d-\frac{1}{2^s}, d+\frac{1}{2^s})$ into $O$.  That finishes the algorithm.

We now verify that this algorithm works.  If $x$ gets enumerated, say in the second step via witness $\sigma$.
Suppose $f(x)\downarrow$ and $\alpha$ is a Cauchy representative for $x$, then $M^{\alpha}_e\downarrow$.  By the fine-tune procedure,
$\sigma$ is an initial segment of an infinite path on the tree $T_{\alpha}$, thus $M^{\alpha}_e\neq 0$,
because its first $j$ bits are incompatible with $0^j$.  Similarly for those $x$ enumerated in the third step.

On the other hand, suppose $f(x)\neq 0$.  We consider the following two cases:
Case 1. $x$ is nondyadic.  Then we know there is a unique dyadic sequence $\delta$ such that $\delta(i)$ is the biggest dyadic in $D_i$
that is less than $x$.
So $M_e^{\delta}$ must be incompatible with $0$. By the stage that we discover the fact, $x$ is enumerated into $O$.
Case 2. $x$ is a dyadic, say $d$.  Then $x$ is enumerated into $O$ in step 3 by a similar argument.
\end{proof}

\begin{defi}
We say that a subset $A\subseteq \R$ is {\em recursive over $\R$}
if its characteristic function $\chi_{_A}$ is partial recursive over $\R$.
\end{defi}

As in Baire space, recursive over $\R$ sets are closed under complementation.

\begin{lem} \label{lem:open_R}
Every effectively open subset $A$ of $\R$ is recursive over $\R$. So is every effectively closed set.
\end{lem}

Unlike in $\baire$, where the membership of a basic open set $\nbhd{\sigma}$ can be decided by a single slave,
in $\R$, it needs the collective effort of all slaves and a zero-test to determine the membership of $B(c,r)$ (see Claim 2 below).
That is why the proof is much more complicated than the proof of Lemma \ref{lem:pi01isstronglycomputable},
which is its counterpart over $\baire$.

\begin{proof}
For a rational interval $I=(p,q)$, we say that $L$ is a {\em linear function} on $I$ if the domain of $L$ is $[p,q]$ and $L(x)=\frac{s-r}{q-p}(x-p)+r$ for some rationals $s$ and $t$,
in other words, it is a line connecting $(p,r)$ and $(q,s)$ on the plane.
We say that $f$ is a {\em piecewise linear function},
if there are finitely many consecutive rational intervals $[p_0,p_1],\dots,[p_{n-2}, p_{n-1}]$ ($n\geq 2$)
such that $f\restriction [p_k,p_{k+1}]$ is a linear function $L_k$ and $L_k(p_{k+1})=L_{k+1}(p_{k+1})$.

{\bf Claim 1}. Any piecewise linear function $f$ is TTE computable.

{\bf Proof Sketch of Claim 1}. Given any input $\alpha$ representing $x\in \R$, let $y_n= f(\alpha(n))$.
Since the endpoints of the intervals are rational, it is recursive to decide $\alpha(n)$ belongs to which interval(s),
so that we know which linear function to apply.
And the continuity of $f$ ensures that if $\alpha\sim \beta$ then $(f(\alpha(n))\sim f(\beta(n))$.

By Claim 1, we can show that

{\bf Claim 2}. Any single rational interval $I=(p,q)$ is MS-computable.

{\bf Proof Sketch of Claim 2}. Consider the following ``zero-test preparing function'' for $I$
\[
f(x)=\left\{
       \begin{array}{ll}
         0, & \hbox{if $x\leq p$ or $x\geq q$;} \\
         \frac{1}{q-p}(x-p), & \hbox{if $p\leq x\leq \frac{p+q}{2}$.}\\
         -\frac{1}{q-p}(x-\frac{p+q}{2})+\frac{1}{2}, &\hbox{if $\frac{p+q}{2}\leq x\leq q$}
       \end{array}
     \right.
\]
which is TTE-computable by Claim 1.  So we let the slaves to compute $f(x)$ and write the result on the zero-test tape.
Since $x\in I$ if and only if $f(x)\neq 0$, by invoking the zero-test, we know if $x\in I$.

Note that the preparing function for a single open interval can be chosen to be the distance function $d_C$,
where $C$ is the closed set $\R\setminus I$.  But for open sets in general, the distance function will not work.
But luckily we do not require to find the distance if we just want to determine the membership of the open set.

{\bf Claim 3}.  Any effectively open subset $U$ of $\R$ is MS-computable.

{\bf Proof of Claim 3}.
Now let $U=\bigcup_n I_n$ be an effective open set and fix a recursive enumeration $(I_n)$.
Without loss of generality, we assume that each $I_n$ is bounded and we allow repetitions,
so that we can assume that $I_n$ is enumerated at stage $n$.  Since we cannot afford to have infinitely many zero-tests,
we need is a preparing function $f$ which works for the whole $U$.
In other words, we need a function $f$ satisfying (1) $f$ is TTE computable; (2) $x\in  U$ if and only if $f(x)\neq 0$;
(3) the algorithm for computing $f(x)$ does not depend on the presentation $\alpha$ of $x$.

We will define a sequence $(f_n)$ of piecewise linear functions and $f$ will be the pointwise limit of $f_n$.

Let $f_0$ be the zero-test preparing function for $I_0$ as described in the proof of Claim 2.
Suppose that $f_n$ have been defined and $U_n =\{x: f_n(x)\neq 0\}$ which is a finite union of rational intervals.
Look at $I_{n+1}$.

Case 1. $I_{n+1}\subseteq U_n$, then let $f_{n+1}=f_n$.

Case 2.  $I_{n+1}\cap U_n=\emptyset$, then let
\[
f_{n+1}(x)=\left\{
\begin{array}{ll}
f_n, & \hbox{if $x \notin I_{n+1}$;}\\
\frac{1}{2^n} g_n(x), & \hbox{if $x\in I_{n+1}$}
\end{array}\right.
\]
where $g_n(x)$ is the zero-test preparing function for $I_{n+1}$ defined as in the proof of Claim 2.
The factor $\varepsilon=\frac{1}{2^n}$ is needed,
because otherwise we may have some non-Cauchy sequence $y\in \Q^{\omega}$ with $y_s=0$ and $y_{s+1}>\frac{1}{2^s}$ as the outcome.

Case 3. Neither Case 1 nor Case 2.  Then define
\[
f_{n+1}(x)=\left\{
\begin{array}{ll}
f_n, & \hbox{if $x \notin I_{n+1}$;}\\
h_n(x), & \hbox{if $x\in I_{n+1}$}
\end{array}\right.
\]
where $h_n$ can be viewed as a ``$\frac{1}{2^n}$ lifting of $f_n$ within $I_{n+1}$''.
To be more precise, let $f$ be a piecewise linear function, $I=(p,q)$ be a rational interval and $\varepsilon>0$ is a rational number,
we say that $h:[p,q]\to \R$ is an $\varepsilon$-lifting of $f$ within $I$, if $h$ is obtained from $f$ as follows:
Let $p\leq q_0< \dots<q_k\leq q$ such that $(q_i,q_{i+1})$ are the domains of all linear pieces of $f$ between $p$ and $q$.
Let $r=\max\{0,f(p)\}$, $s_i=f(q_{i+1})+\varepsilon$ and $s=\max\{0,f(q)\}$.
Then $h$ is the piecewise linear function connecting $(p,r),(q_1,s_1), \dots, (q_k,s_k), (q,s)$ on $\R^2$.

Since $|f_{n+1}(x)-f_n(x)|\leq \frac{1}{2^n}$ for all $x\in \R$, the sequence $(f_{n})$ converges to $f$ uniformly.
Consequently $f$ is continuous.  The definition above induces a TTE-procedure to compute $f$:
Each slave $S_i$ just compute $f_i(\alpha(i))$.  Since the definition of $f$ does not depend on $\alpha$ and $f$ is continuous,
this procedure is well-defined and Cauchy preserving.
Finally, $x\in U$ iff $x\in U_i$ for some $i$ iff $f(x)\neq 0$.

In summary, $U$ can be computed using two master steps: Step 1: make the preparing function $f$ and write $f(\alpha)$ on the zero-test tape; Step 2: execute zero-test and read off the answer.
\end{proof}

\begin{thm} \label{thm:rec_in_R}
A subset $A$ of $\R$ is recursive over $\R$ if and only if it is $\Delta^0_2$ over $\R$.
\end{thm}

\begin{proof}  ``($\Leftarrow$)''
By the proof of Lemma \ref{lem:open_R}, it takes only two master steps to determine if $x$ is in a basic open set.
So the whole proof of Lemma \ref{lem:delta02isstronglycomputable} can go through in the context of $\R$.

``($\Rightarrow$)'' By Lemma \ref{lem:finding_index_for nonzero}, we can uniformly find the indices of $\{x\in\R: \Phi^x(e)\neq 0\}$.
So the whole proof of Theorem \ref{thm:stronglycomputableisdelta02} can go through in the context of $\R$.
\end{proof}

\section{Properties of MS-computable functions and sets}

\subsection{The decomposibility of MS-computable functions}

Analyzing the proofs of Lemma \ref{lem:delta02isstronglycomputable} and Theorem \ref{thm:stronglycomputableisdelta02},
we see that the same analysis holds for any total MS-computable function (rather than the characteristic function of a set) over $\mathcal{N}$ or $\R$.
In fact, one can derive the following:

\begin{cor}\label{cor:decomposability}
Let $F:\R\rightarrow\R$. Then $F$ is (total) MS-computable if and only if there is an effective partition $\{X_i\}$ of $\R$ into $\Delta^0_2$ sets
and an effective sequence of continuous functions $\{H_i\}$ such that for every $i$, $F\upharpoonright X_i=H_i\upharpoonright X_i$.

Furthermore, each $X_i$ can be taken to be the intersection of an effectively open set and an effectively closed set.
\end{cor}

Our investigation is related to a well-known result in descriptive set theory, the Jayne-Rogers theorem.
We can state the version for $\R$ as the following:

\begin{thm}[Jayne, Rogers \cite{JR1982}]
Let $f:\R\rightarrow \R$. Then $f^{-1}(O)$ is $\Delta^0_2$ for every open set $O$ if and only if
there is a partition of $\R$ into $\Delta^0_2$ sets $\{X_i\}$ such that $f\upharpoonright X_i$ is continuous for each $i$.
\end{thm}

Effective versions of the Jayne-Rogers theorem were also studied and discussed in Pauly, de Brecht \cite{paulybrecht}.
Pauly and de Brecht showed that a certain effective version of the Jayne-Rogers theorem was true for computable metric spaces:

\begin{thm} [Pauly, de Brecht \cite{paulybrecht}]
A function $f:\mathbb{R}\to\R$ is effectively $\Delta^0_2$-measurable if and only if it is piecewise computable.
\end{thm}

We note that the notion of being piecewise computable is the same as the existence of the decomposition in Corollary \ref{cor:decomposability}. Thus, we obtain:

\begin{cor}
A function $f:\mathbb{R}\to\R$ is MS-computable if and only if it is effectively $\Delta^0_2$-measurable if and only if it is piecewise computable.
\end{cor}

\subsection{Comparing TTE and MS-computability over $\baire$}

We now give some examples of nonrecursive sets over $\baire$ and $\R$.

\begin{prop}
  The set $A=\{x\in \baire: x$ has infinitely many zeros$\}$ is not MS-computable over $\baire$.
\end{prop}

\begin{proof}
  We reduce the $\Pi^0_2$-complete set of natural numbers Inf$=\{e: W_e$ is infinite$\}$ to the set $A$.

  Suppose that $A$ is computed by the MS-machine $M$.  We convert $M$ to get another machine $M_I$ to compute Inf as follows:
  For any input $e\in \N$, first ask each slave $S_i$ to check if $W_{e,i}$ has more elements than $W_{e,i-1}$.  If the answer is yes,
  then write a zero on the zero-test tape, otherwise write a one.
  Now apply $M$ to get an answer, which will be the outcome of $M_I$. Clearly, $W_e$ is infinite if and only if the zero-test tape has infinitely many zeros.
  So $M_I$ correctly computes Inf, contradict Proposition \ref{char_on_N}.
\end{proof}

\begin{prop}
  The set $\Q$ is not MS-computable over $\R$.
\end{prop}

\begin{proof}
  Because $\Q$ is not a $G_{\delta}$ set, the result follows from Theorem \ref{thm:rec_in_R}.
\end{proof}

\begin{prop}
If $x\in \baire$ and $M(x)\downarrow=y$, then there is a TTE-computable function $f$ such that $f(x)=y$,
i.e., there is a Turing functional $\lambda x.\Phi^x$ such that $\Phi^x\downarrow=y$.
In particular, if $x$ is a total recursive function over $\N$ and $M(x)=y$, then $y$ is a total recursive function over $\N$.
However the map from the index of $x$ to the index of $y$ is $\emptyset'$-recursive.
\end{prop}

\begin{proof}
(sketch) During the computation of $M(x)=y$, $M$ only uses finitely many zero-test steps.
We may code the answers of the zero test results as a single parameter, then the computation becomes TTE.
\end{proof}

However, there are some subtleties.
\begin{prop} \label{prop:subtleties}
\begin{enumerate}
\item[(a)] The singleton set $\{\chi_{\emptyset'}\}$ is not recursive over $\baire$ (as a subset of $\baire$),
even though $\emptyset'$ is recursive over $\baire$ (as a subset of $\omega$).
\item [(b)] Let $f\in \baire$ be such that $f(n)=0$ if $n\notin \emptyset'$,
and $f(n)=s$ if $n\in \emptyset'$ and $s$ is the least stage of $n$ entering $\emptyset'$.
Then $\{f\}$ is recursive over $\baire$. (This $\{f\}$ is a $\Pi^0_1$ singleton coding $\emptyset'$ in $\baire$.)
\end{enumerate}
\end{prop}

\begin{proof}
(a)  Suppose that $\{\chi_{\emptyset'}\}$ is recursive over $\baire$.  Then by Theorem \ref{thm:stronglycomputableisdelta02},
$\{\chi_{\emptyset'}\}$ is a $\Sigma^0_2$-singleton, hence it is a $\Pi^0_1$-singleton (by taking the $\Sigma_2^0$-witness as a parameter) in the Baire space.  In other words,
there is a recursive tree $S$, with $\chi_{\emptyset'}$ as its unique infinite branch. However, $\chi_{\emptyset'}$ is $0$-$1$ valued, so by
restricting $S$ to the Cantor space, we get a recursive binary tree $T$ with $\chi_{\emptyset'}$ as its unique infinite branch, which is obviously impossible.

(b) Let $M$ be the following MS-machine: For any input $x$, let the $i$-th slave check the correctness of the first $i$ bits of $x$ up to step $i$.
To be more precise, for each $j<i$, if $x(j)=0$ then check whether $j\notin W_{j,i}$; if $x(j)=s\neq 0$, then check if $s$ is the least stage that $j\in W_{e,s}$.
If the $i$-slave found an error, then writes a one on the zero-test tape, otherwise, writes a zero.  Clearly $x=f$ if and only if the sequence written on the zero-test tape is
the zero sequence.  By applying zero-test, the master knows if $x=f$.
\end{proof}

Proposition \ref{prop:subtleties} tells us: (1) A set $A\subseteq \omega$ being a recursive set is different from its characteristic function $\chi_{_A}$ being a MS-computable singleton in $\baire$.
(2) It is possible to have two elements $x$ and $y$ in the Baire space, which are TTE-equivalent, i.e., there are Turing functionals $\Phi$ and $\Psi$ with $\Phi^x=y$ and $\Psi^y=x$,
but $\{x\}$ is MS-computable, whereas $\{y\}$ is not.

\subsection{Comparing MS-computability over $\baire$ and over $\R$}

Despite the computable sets in $\baire$ and in $\R$ being $\Delta^0_2$-definable within the structure,
their differences are striking. One of the reasons is that $\baire$ is not locally compact whereas $\R$ is.

The following folklore is well-known:
\begin{fact} \label{fact:hyp}
For each recursive ordinal $\alpha$, there is some $f\in\baire$ such that $f\equiv_T\emptyset^{(\alpha)}$ and $f$ is a $\Pi^0_1$-singleton.
\end{fact}

\begin{proof}
By Sacks \cite[Theorem II.4.2]{Sac90}, for each computable ordinal $\alpha$, $\emptyset^{(\alpha)}$ is a $\Pi^0_2$-singleton.
By Jockusch and McLaughlin, \cite[Theorem 3.1]{JM1969}, we can replace each $\Pi^0_2$-singleton with a Turing equivalent $\Pi^0_1$-singleton.
\end{proof}

We compare the singletons that are MS-computable in $\baire$ and in $\mathbb{R}$.
Each MS-computable singleton in $\mathbb{R}$ must be a recursive real, and by Fact \ref{fact:hyp},
the MS-computable singletons in $\baire$ are far from being recursive:

\begin{lem} \label{lem:singleton}
\begin{itemize}
\item [(i)] $\{f\}\subseteq\baire$ is MS-computable if and only if $f$ is a $\Pi^0_1$-singleton.
\item [(ii)] $\{x\}\subseteq\mathbb{R}$ is MS-computable if and only if $x$ is a recursive real.
\end{itemize}
\end{lem}

\begin{proof}
(i): Each $\Pi^0_1$-class in $\baire$ is MS-computable by Lemma \ref{lem:delta02isstronglycomputable}.
On the other hand, if $\{f\}$ is MS-computable then it is a $\Delta^0_2$ class by Theorem \ref{thm:stronglycomputableisdelta02},
and each $\Sigma^0_2$-singleton is also $\Pi^0_1$.

(ii): The right-to-left direction is obvious. Now suppose that $\{x\}$ is MS-computable, hence it is effectively $\Delta^0_2$,
and hence effectively closed. However, each effectively closed singleton in $\mathbb{R}$ is clearly both left-r.e.~and right-r.e.

\end{proof}

A set $X\subseteq Y$ is said to {\em determine} a class of functions with domain $Y$,
if for any two functions $F$ and $G$ from this class, $F\upharpoonright X=G\upharpoonright Y$ implies that $F=G$.
For instance, the class of TTE-computable functions is determined by the class of dyadic rationals.

Next, we compare MS-computability on $\baire$ and on $\mathbb{R}$ by comparing the least complicated set that determines MS-computability.
Again, we see that MS-computable functions on $\mathbb{R}$ can be represented much more simply than MS-computable functions on $\baire$:

\begin{prop}
\begin{itemize}
\item [(i)] The class of MS-computable functions on $\baire$ is determined by the class of all $f\in\baire$ with $f\leq_T \mathcal{O}$ (the Kleene's $\mathcal{O}$).
Furthermore, for any recursive ordinal $\alpha$, the class $\Delta^0_\alpha$ does not determine the class of MS-computable functions on $\baire$.
\item[(ii)] The class of MS-computable functions on $\mathbb{R}$ is determined by the class of all $\Delta^0_2$ reals (in fact, low reals), but not by the class of all computable reals.
\end{itemize}
\end{prop}

\begin{proof}
(i): If $F$ and $G$ are MS-computable, then $F(\alpha)\leq_T\alpha'$, and thus $F(\alpha)\neq G(\alpha)$ is an arithmetical predicate.
By Kleene's Basis Theorem (see \cite[Theorem 1.3 on Chapter 3]{Sac90}), this is witnessed by some $\alpha\leq_T \mathcal{O}$.
On the other hand, by Fact \ref{fact:hyp} and Lemma \ref{lem:singleton}, for any recursive ordinal $\alpha$,
there are MS-machines $M_d$ and $M_e$ computing the sets $\{F\}$ and $\{G\}$ for some $F\equiv_T\emptyset^{(\alpha+1)}$ and some $G\equiv_T\emptyset^{(\alpha+2)}$.
So they differ only at the points $F$ and $G$ and in particular, they are equal on all inputs recursive in $\emptyset^{(\alpha)}$.

(ii): By Corollary \ref{cor:decomposability},
if $F$ and $\hat{F}$ are MS-computable and $F(x)\neq \hat{F}(x)$ for some real $x$,
then for some $\Delta^0_2$ sets $X,\hat{X}$, and some effectively continuous $H$ and $\hat{H}$,
we have $x\in X\cap \hat{X}$, $F(x)=H(x)$ and $\hat{F}(x)=\hat{H}(x)$. Since $H(x)\neq\hat{H}(x)$,
let $I$ be a small enough rational interval containing $x$ such that $\left(H\upharpoonright I\right)\cap \left(\hat{H}\upharpoonright I\right)=\emptyset$.
Now consider $X\cap\hat{X}\cap I$; this is a non-empty $\Delta^0_2$-class in $\mathbb{R}$,
and therefore must contain some low member, say $\hat{x}$.
But this means that $F(\hat{x})=H(\hat{x})\neq \hat{H}(\hat{x})=\hat{F}(\hat{x})$.  On the other hand,
there is an effectively closed set $C$ which has no recursive member,
any MS-machine which computes $C$ will agree with constant zero function on all recursive real.
\end{proof}

\section{Concluding Remarks}
There is an interesting debate about whether algorithms should be defined in terms of abstract machines
(Gurevich \cite{Gurevich:2000} \cite{Gurevich:2012}) or in terms of ``recursor'' (Moschovakis \cite{Moschovakis:2001}).
Vardi made some inspiring remarks in his short article \cite{Vardi:2012}.  He mentioned the de Broglie's {\em wave particle duality}
in physics and claimed
\begin{quote}
  An algorithm is both an abstract state machine and a recursor, and neither view by itself fully describes what an algorithm is.
This algorithmic duality seems to be a fundamental principle of computer science.
\end{quote}
This duality occurred in the classical definition of algorithms over natural numbers, and it occurred again in our analysis of
computability over real numbers.  Our analysis seemed to suggest that the classical correspondence between the recursion schemes
and Turing machines can be applied to other domains.
Both the recursion scheme and the Turing machine can be viewed as control units above some domain-dependent primitive functions or operations.

However, our work heavily depends on the underlying algebraic and topological structure of real numbers,
and the machine part also relies on representing a real number as an $\omega$-sequence.
Comparing to the natural and elegant formalizations of computability over natural numbers,
in particular the works by G\"odel \cite{Godel:1934} and Turing \cite{Turing:1936}, more insights are needed.
The ultimate question remains to be: Is there a natural, elegant and general definition of algorithms over all domains?  If so, what is it?

\end{document}